\documentclass[11pt]{article}
\usepackage{amsfonts}
\usepackage{amssymb}
\usepackage{amsmath}
\usepackage{amsthm}
\usepackage{geometry}
\usepackage{commath}
\usepackage{booktabs}
\usepackage{lscape}
\usepackage{float}
\usepackage{xr}

\usepackage[authoryear]{natbib}

\newcommand{\myref}[2]{\hyperref[#1]{#2}}
\numberwithin{equation}{section}

\usepackage{latexsym}
\usepackage{graphicx} % include figures
\usepackage{enumerate}
\usepackage{setspace}
\usepackage[toc,title,titletoc,header]{appendix}
\usepackage[bottom]{footmisc}  % forces the footnotes to go at the bottom of page.
\usepackage[pdfborder={0 0 0}]{hyperref}
\usepackage{multicol}
\usepackage{multirow}

\usepackage{lineno}
\setpagewiselinenumbers
\newtheorem{theorem}{Theorem}[section]
\newtheorem{lemma}{Lemma}[section]

\theoremstyle{definition}

\newtheorem{remark}{Remark}[section]

\theoremstyle{remark}

\newcounter{assumptionM}
\newcounter{assumptionA}
\def\theassumptionM{M.\arabic{assumptionM}}
\def\theassumptionA{\arabic{assumptionA}}
\newenvironment{assumptionA}[1][]{\refstepcounter{assumptionA}\medskip\noindent{\textbf{Assumption \theassumptionA. #1}}}{\medskip}

\hypersetup{
  colorlinks=true,linkcolor=blue, citecolor=blue, filecolor=blue, 
  linkbordercolor={0 1 1}, citebordercolor={1 0 0},
  pdftitle={Inference in dynamic discrete choice problems under local misspecification},
  pdfauthor={Bugni and Ura},
  pdfcreator={\LaTeX\ with package \flqq hyperref\frqq},
  pdfsubject={},
}

\defcitealias{aguirregabiria/mira:2002}{AM}

\begin{document}
\sloppy

%-------- Beginning Title Page -----------------------------
\hypersetup{pageanchor=false}
\title{\Large Inference in dynamic discrete choice problems under local misspecification\thanks{
We thank the Editor, Chris Taber, and three anonymous referees for comments and suggestions that have significantly improved this paper. We are also grateful for helpful discussions with Victor Aguirregabiria, Peter Arcidiacono, Joe Hotz, Shakeeb Khan, Matt Masten, Arnaud Maurel, Jia Li, and the participants at the Duke Microeconometrics Reading Group and the Yale Econometrics Lunch Group. Any and all errors are our own. The research of the first author was supported by NIH Grant 40-4153-00-0-85-399 and NSF Grant SES-1729280.}
}
\author{Federico A. Bugni\\Department of Economics\\Duke University\\ \href{mailto:federico.bugni@duke.edu}{\texttt{federico.bugni@duke.edu}}
\and Takuya Ura\\Department of Economics\\University of California, Davis\\ \href{mailto:takura@ucdavis.edu}{\texttt{takura@ucdavis.edu}}} 
\date{\today{\\% \vspace{0.5cm}PRELIMINARY AND INCOMPLETE\\PLEASE DO NOT DISTRIBUTE WITHOUT PERMISSION
}}
\maketitle

\begin{abstract}
	
	Single-agent dynamic discrete choice models are typically estimated using heavily parametrized econometric frameworks, making them susceptible to model misspecification. This paper investigates how misspecification affects the results of inference in these models. Specifically, we consider a local misspecification framework in which specification errors are assumed to vanish at an arbitrary and unknown rate with the sample size. Relative to global misspecification, the local misspecification analysis has two important advantages. First, it yields tractable and general results. Second, it allows us to focus on parameters with structural interpretation, instead of ``pseudo-true'' parameters.
	
	We consider a general class of two-step estimators based on the $K$-stage sequential policy function iteration algorithm, where $K$ denotes the number of iterations employed in the estimation. This class includes \cite{hotz/miller:1993}'s conditional choice probability estimator, \cite{aguirregabiria/mira:2002}'s pseudo-likelihood estimator, and \cite{pesendorfer/schmidt-dengler:2008}'s asymptotic least squares estimator.

We show that local misspecification can affect the asymptotic distribution and even the rate of convergence of these estimators. In principle, one might expect that the effect of the local misspecification could change with the number of iterations $K$. One of our main findings is that this is not the case, i.e., the effect of local misspecification is invariant to $K$. In practice, this means that researchers cannot eliminate or even alleviate problems of model misspecification by changing $K$.

\begin{description}
	\item[Keywords:] Single-agent dynamic discrete choice models, estimation, inference, misspecification, local misspecification.
	\item[JEL Classification Codes: C13, C61, C73]
\end{description}
\end{abstract}
\vfill

\thispagestyle{empty} % no page number \clearpage
\pagebreak % The paper begins in a new page
\setcounter{page}1 % Set the next page number to 1
% \tableofcontents
\pagebreak % The paper begins in a new page
%-------- End Title Page -----------------------------
\hypersetup{pageanchor=true}

\section{Introduction}

This paper investigates the effect of model misspecification on inference in single-agent dynamic discrete choice models. Our study is motivated by two observations regarding this literature.\footnote{See survey papers by \cite{aguirregabiria/mira:2010} and \cite{arcidiacono/ellickson:2011}, and references therein.} First, typical econometric frameworks used in empirical studies are heavily parametrized and are therefore subject to misspecification. Second, several methods can be used to estimate these models, including \citet{rust:1987,rust:1988}'s nested fixed point estimator, \cite{hotz/miller:1993}'s conditional choice probability estimator, \cite{aguirregabiria/mira:2002}'s pseudo-likelihood estimator, and \cite{pesendorfer/schmidt-dengler:2008}'s asymptotic least squares estimator. While the literature has studied the behavior of these estimators under correct specification, their properties under misspecification have not been explored. To the best of our knowledge, our paper is one of the first ones to investigate the effect of misspecification on inference in these types of models.

In this paper, we propose a local misspecification approach in dynamic discrete choice models. By local misspecification, we mean that the econometric model is allowed to be misspecified, but the amount of misspecification vanishes as the sample size increases. Local misspecification is an asymptotic device that can provide concrete conclusions in the presence of misspecification while keeping the analysis tractable. As with any other asymptotic device, local misspecification is just an approximation to a finite sample situation (in this case, with misspecification) and should not be taken literally.

Our local approach to misspecification in dynamic discrete choice models yields relevant conclusions in several dimensions. First, local misspecification can constitute a reasonable approximation to the asymptotic behavior when the mistakes in the specification of the model are small. Second, there are multiple available estimation methods, and their performance under misspecification is not well understood. We believe their relative performance under local misspecification is a relevant comparison criterion. Finally, while our approach to misspecification is admittedly ``local'' in that it is assumed to vanish, we allow the rate at which this occurs to be completely arbitrary. In particular, we allow the misspecification to disappear at a faster, equal, or even slower rate than the regular parametric convergence rate of $\sqrt{n}$. While this rate will affect the asymptotic properties of the estimators under consideration, it will not alter the main qualitative conclusions of our paper.

% \footnote{The asymptotic distribution of the estimators under consideration will depend on the convergence rate of the misspecification. If the misspecification disappears at a faster rate than $\sqrt{n}$, then the asymptotic distribution will coincide with the one obtained under correct specification. If the misspecification vanishes exactly at the $\sqrt{n}$ rate, then the asymptotic distribution will be the one obtained under correct specification plus an additional bias term. Finally, if the misspecification disappears at a slower rate than $\sqrt{n}$, then the rate of convergence of the estimators will change and their asymptotic distribution will be degenerate and equal to the previously mentioned bias term.}

We consider a class of two-step estimators based on the $K$-stage sequential policy function iteration algorithm along the lines of \cite{aguirregabiria/mira:2002}, where $K$ denotes the number of iterations employed in the estimation. By appropriate choice of the criterion function, this class captures the $K$-stage maximum likelihood type estimators ($K$-ML) and the $K$-stage minimum distance type estimators ($K$-MD). This class includes most of the previously mentioned estimators as special cases.

Our main theoretical contribution is to characterize the asymptotic distribution of two-step $K$-MD and $K$-ML estimators under local misspecification. We show that local misspecification can affect the asymptotic distribution and even the rate of convergence of these estimators. We are particularly interested in the asymptotic behavior of these estimators as we vary the number of iterations $K$.

We obtain three main results. Our first result is related to the asymptotic behavior of $K$-ML estimators. Under correct specification, \cite{aguirregabiria/mira:2002} proves that the asymptotic distribution of $K$-ML estimators is invariant to $K$. Under local misspecification, however, one might reasonably expect a different result. Intuitively, every stage of the policy function iteration algorithm brings the estimator closer to imposing the fixed-point/equilibrium conditions implied by the model. Given that the model is incorrectly specified, one might then conjecture that increasing the number of iterations would result in an estimator of inferior quality (e.g.\ more bias).\footnote{We are grateful to an anonymous referee for suggesting this interpretation.} Our first main result is to show that this intuition is incorrect. We formally show that $K$-ML estimators are asymptotically equivalent for all $K$. Our second result is to show an analogous result for $K$-MD estimators, i.e., given the choice of weight matrix, $K$-MD estimators are asymptotically equivalent for all $K$. 
If we combine these findings, we can conclude that the researcher cannot eliminate or even alleviate a problem of model misspecification by changing the number of iterations $K$. Additional iterations are computationally costly and produce {\it no change} in asymptotic efficiency. Thus, from a practical viewpoint, we recommend using either the $1$-ML or the $1$-MD estimator.

Finally, our third result is to compare $K$-MD and $K$-ML estimators in terms of asymptotic mean squared error. We show that an optimally-weighted $K$-MD estimator depends on the unknown asymptotic bias and is thus generally unfeasible. In turn, the feasible $K$-MD estimator with a weight matrix that minimizes asymptotic variance could have an asymptotic mean squared error that is higher or lower than that of the $K$-ML estimator or the $K$-MD estimator with identity weight matrix. In other words, given a particular choice of the number of iterations $K$ (e.g.\ $K=1$), the presence of local misspecification implies that we cannot make clear-cut recommendations regarding the weight matrix for the $K$-MD estimator, and how this compares with the $K$-ML estimator.

From a technical viewpoint, our analysis exploits a distinctive feature of single-agent dynamic discrete choice problems known as the ``zero Jacobian property''. This property was used by \cite{aguirregabiria/mira:2002} under correct specification to their results for $K$-ML estimators. One of our technical contributions is to use this property under local misspecification to derive analogous results for both $K$-ML and $K$-MD estimators. 

As we have explained, this paper uses a local approach to the problem of model misspecification. In practice, researchers typically specify econometric models that may contain non-vanishing errors, i.e., global misspecification. Relative to the global misspecification analysis, our local misspecification approach has two important advantages. First, allowing for global misspecification in our dynamic discrete choice model typically makes the problem intractable, and generally valid results are thus hard to obtain. In contrast, the local misspecification yields concrete and general conclusions. Second, recall that the literature has produced several estimation methods to estimate the structural parameter of interest in a dynamic discrete choice problem. Under global misspecification, the different estimators typically converge in probability to different pseudo-true parameters which may or may not be related to the true structural parameter. This makes the results hard to interpret and compare. In contrast, under local misspecification, these different estimators are shown to consistently estimate the true structural parameter value. We can then compare their robustness to misspecification via their asymptotic distributions.

This paper relates to a vast literature on inference under model misspecification. \cite{white:1982,white:1996} consider the problem of maximum likelihood estimation under global misspecification. \cite{newey:1985a,newey:1985b} and \cite{tauchen:1985} investigate the power properties of the model specification tests under local misspecification. More recently, \cite{schorfheide:2005} considers a locally misspecified vector autoregression process and proposes an information criterion for the lag length in the autoregression model. \cite{bugni/canay/guggenberger:2012} compares inference methods in partially identified moment (in)equality models that are locally misspecified. \cite{kitamura/otsu/evdokimov:2013} considers a class of estimators that are robust to local misspecification in the context of moment condition models. None of the previously mentioned references consider inference in dynamic discrete choice problems, whose specific features are central to the results in this paper. Few references explore the issue of misspecification in dynamic discrete choice problems. For example, \cite{norets/takahashi:2013} considers surjective dynamic discrete choice models and show that these are necessarily correctly specified. Last, \cite{chernozhukov/escanciano/ichimura/newey:2016} considers two-step estimators in dynamic discrete choice model with a locally misspecified first step, and propose estimators that are robust to this issue. In contrast, this paper allows both steps to be locally misspecified.

The remainder of the paper is structured as follows. Section \ref{sec:Setup} describes the dynamic discrete choice model and introduces the possibility of its local misspecification. Section \ref{sec:Inference} develops a general result for two-step $K$-stage estimators under high-level conditions. Section \ref{sec:Applications} applies the general result to $K$-ML estimators (Section \ref{sec:ML}) and $K$-MD estimators (Section \ref{sec:MD}). Section \ref{sec:MonteCarlos} presents results of Monte Carlo simulation and Section \ref{sec:Conclusion} concludes. The appendix of the paper collects all the proofs and intermediate results.

The following notation is used throughout the paper. For any $s_1,s_2\in \mathbb{N}$, $\mathbf{0}_{s_1 \times s_2}$ and $\mathbf{1}_{s_1 \times s_2}$ denote a $(s_1 \times s_2)$-dimensional matrix composed of zeros and ones, respectively, and $ \mathbf{I}_{s_1 \times s_2}$ denotes $(s_1 \times s_2)$-dimensional matrix equal to the left upper block of the $(\max(s_1,s_2) \times \max(s_1,s_2))$-dimensional identity matrix. We use $||\cdot||$ to denote the Euclidean norm. For any $s$-dimensional column vector $V$, $diag\{V\}$ is the $(s \times s)$-dimensional matrix with $V$ as its diagonal. For sets of finite indices $S_1 = \{1,\dots,|S_1|\}$ and $S_2 = \{1,\dots,|S_2|\}$, $\{M(s_1,s_2)\}_{(s_1,s_2)\in S_1\times S_2} $ denotes the $(|S_1|\times |S_2|)$-dimensional column vector equal to the vectorization of $ \{\{M(s_1,s_2)\}_{s_1=1}^{|S_1|}\}_{s_2=1}^{|S_2|}$. For any differentiable matrix function $F(y):\mathbb{R}^{a \times b} \to \mathbb{R}^{c \times d}$, $\partial F(y)/\partial y \in \mathbb{R}^{ac \times bd}$ denotes the usual matrix of derivatives. Finally, ``w.p.a.1'' abbreviates ``with probability approaching one''.

\section{Setup}\label{sec:Setup}

Section \ref{sec:Model} describes the dynamic discrete choice model assumed by the researcher. This paper allows this model to be incorrectly specified. Section \ref{sec:Misspecification} describes the nature of the model misspecification.

\subsection{The econometric model}\label{sec:Model}

An economic agent is assumed to behave according to the discrete Markov decision framework in \cite{aguirregabiria/mira:2002}. In each period $t =1,\dots,T \equiv \infty$, the agent is assumed to observe a vector of state variables $s_{t}$ and to choose an action $a_{t}\in A \equiv \{1,\ldots ,|A|\}$ with the objective of maximizing the expected discounted utility. The vector of state variables $s_{t}=(x_{t},\epsilon_{t})$ is composed by two subvectors. The subvector $x_{t}\in X \equiv \{1,\dots ,|X|\}$ represents a scalar state variables observed by the agent and the researcher, whereas the subvector $\epsilon_{t}\in \mathbb{R} ^{|A|}$ represents an action-specific state vector only observed by the agent.

The agent's future state variables $(x_{t+1},\epsilon_{t+1})$ are assumed to follow a Markov transition probability density $d\Pr(x_{t+1},\epsilon_{t+1}| x_{t},\epsilon_{t},a_{t})$ that satisfies:
\begin{equation*}
d\Pr(x_{t+1},\epsilon_{t+1}| x_{t},\epsilon_{t},a_{t}) ~=~g_{\theta _{g}}(\epsilon_{t+1}|x_{t+1})f_{\theta _{f}}(x_{t+1}|x_{t},a_{t}),
\end{equation*}
where $g_{\theta _{g}}(\cdot)$ is the (conditional) distribution of the unobserved state variable with parameter $\theta_{g}$ and $f_{\theta _{f}}(\cdot)$ is the transition probability of the observed state variable with parameter $\theta_{f}$.

The utility is assumed to be time separable and the agent discounts future utility by a known discount factor $\beta \in (0,1)$.\footnote{This follows \citet[Footnote 12]{aguirregabiria/mira:2002} and \cite{magnac/thesmar:2002}.} The current utility function of choosing action $a_{t}$ under state variables $(x_{t},\epsilon_{t})$ is given by:
\begin{equation*}
u_{\theta _{u}}(x_{t},a_{t})~+~\epsilon_{{t}}(a_{t}) ,
\end{equation*}
where $u_{\theta _{u}}(\cdot)$ is non-stochastic component of the current utility with parameter $\theta_{u}$, and $\epsilon_{{t}}(a_{t})$ denotes the $a_{t}$-th coordinate of $\epsilon_{t}$.

The researcher's goal is to estimate the unknown parameters in the model, $\theta  \equiv (\theta _{g},\theta _{u},\theta _{f} )\in \Theta $, where $\Theta$ is the compact parameter space. Also, we denote $\theta =( \alpha ,\theta _{f}) \in \Theta \equiv \Theta _{\alpha }\times \Theta _{f}$ with $\alpha \equiv (\theta_{u},\theta _{g}) \in \Theta _{\alpha }$.

Following \cite{aguirregabiria/mira:2002}, we impose the following regularity conditions on the primitive elements of the econometric model.

\begin{assumptionA}\label{ass:RegularityModel}
For every $\theta \in \Theta$, assume that:
\begin{enumerate}[(a)]
\item For every $x \in X$, $g_{\theta _{g}}(\epsilon|x)$ has finite first moments and is twice differentiable in $\epsilon$,
\item $\epsilon = \{\epsilon(a)\}_{a\in A}$ has full support,
\item $g_{\theta _{g}}(\epsilon|x)$, $f_{\theta _{f}}(x'|x,a)$, and $u_{\theta _{u}}(x,a)$ are twice continuously differentiable with respect to $\theta$.
\end{enumerate}
\end{assumptionA}

By \cite{blackwell:1965}'s theorem and its generalization by \cite{rust:1988}, the optimal decision rule is stationary and Markovian, i.e., the time subscript can be dropped. Furthermore, the optimal value function $V_{\theta }$ is the unique solution of the following Bellman equation:
\begin{equation}
V_{\theta }( x,\epsilon ) =\max_{a\in A}\{ u_{\theta _{u}}( x,a) +\epsilon({a}) +\beta \int_{( x',\epsilon') }V_{\theta }( x',\epsilon') g_{\theta _{g}}(\epsilon'|x')f_{\theta _{f}}(x'|x,a)d( x',\epsilon') \} . \label{eq:BE}
\end{equation}
By integrating out the unobserved error, we obtain the smoothed value function:
\begin{equation*}
V_{\theta }( x) \equiv \int_{\epsilon }V_{\theta }( x,\epsilon ) g_{\theta _{g}}( \epsilon |x) d\epsilon,
\end{equation*}
which is the unique solution of the smoothed Bellman equation:
\begin{equation}
V_{\theta }( x) =\int_{\epsilon } \max_{a\in A}\{ u_{\theta _{u}}( x,a) +\epsilon({a}) +\beta \sum_{x'\in X}V_{\theta }( x') f_{\theta _{f}}( x'|x,a) \} g_{\theta _{g}}( \epsilon |x) d\epsilon. \label{eq:SmoothedBE}
\end{equation}

We now turn to the description of the conditional choice probability (CCP), denoted by $P_{\theta}( a | x) $, which is the model-implied probability that an agent chooses action $a$ when the observed state is $x$. Since the agent chooses an action in $A$, $P_{\theta}(|A||x) = 1- \sum_{a\in \tilde{A}} P_{\theta}(a|x)$ for all $x\in X$. Thus, the vector of model-implied conditional choice probabilities (CCPs) is completely characterized by $P_\theta \equiv \{P_{\theta}(a|x)\}_{(a,x)\in \tilde{A} \times X}$ with $\tilde{A} \equiv  \{1,\dots,|A|-1\}$. For the remainder of the paper, we use $\Theta_{P}\subset [0,1]^{|\tilde{A} \times X|}$ to denote the parameter space for the vector of CCPs.

The vector of CCPs is a central equilibrium object in the model. Lemma \ref{lem:PolicyProperties} shows that the CCPs are the unique fixed point of the policy function mapping. By utility maximization, the vector of CCPs is determined by the following equation:
\begin{equation*}
P_{\theta }( a | x) \equiv \int_{\epsilon } 1\left[ a=\underset{\tilde{a}\in A}{\arg \max} [ u_{\theta _{u}}( x,\tilde{a}) +\beta \sum_{x'\in X}V_{\theta }( x') f_{\theta _{f}}( x'|x, \tilde{a}) +\epsilon(\tilde{a}) ] \right] dg_{\theta _{g}}( \epsilon |x), 
\end{equation*}
which can be succinctly represented as follows:
\begin{equation}
P_\theta ~=~\Lambda _{\theta }( \{ V_{\theta }( x) \} _{x\in X}). \label{eq:Mapping1B}
\end{equation}
Also, notice that Eq.\ \eqref{eq:SmoothedBE} can be re-written as:
\begin{equation}
V_{\theta }( x) =\sum_{a\in A}P_{\theta }(a|x) \left\{ u_{\theta _{u}}( x,a) +E_{\theta }[ \epsilon({a}) |x,a] +\beta \sum_{x'\in X} V_{\theta }( x') f_{\theta _{f}}( x'|x,a) \right\} ,
\label{eq:Mapping2}
\end{equation}
where $E_{\theta }[ \epsilon({a}) |x,a] $ denotes the expectation of the unobservable $\epsilon({a}) $ conditional on the state being $x$ and on the optimal action being $a$. Under our assumptions, \cite{hotz/miller:1993} show that there is a one-to-one mapping between the CCPs and the (normalized) smoothed value function. The inverse of this mapping allows us to re-express $\{E_{\theta }[ \epsilon({a}) |x,a] \}_{(a,x)\in A \times X}$ as a function of the vector of CCPs. By combining this with Eq.\ \eqref{eq:Mapping2}, we can express  $ \{ V_{\theta }( x) \} _{x\in X}$ as a function of $P_\theta$. An explicit formula for such function is provided in \citet[Eq.\ (8)]{aguirregabiria/mira:2002}, which we succinctly express as follows:
\begin{equation}
\{ V_{\theta }( x) \} _{x\in X}~=~\varphi _{\theta }(P_\theta). \label{eq:Mapping2B}
\end{equation}
By combining Eqs.\ \eqref{eq:Mapping1B} and \eqref{eq:Mapping2B}, we obtain the following fixed point representation of the vector of CCPs:
\begin{equation}
P_\theta ~=~\Psi _{\theta }( P_\theta ) , \label{eq:PolicyOperator}
\end{equation}
where $\Psi _{\theta } \equiv \Lambda _{\theta }\circ \varphi _{\theta }$ is the policy function mapping. As explained by \cite{aguirregabiria/mira:2002}, this operator can be evaluated at any vector of CCPs, optimal or not. For any arbitrary $P\equiv \{P(a|x)\}_{(a,x) \in \tilde{A} \times X}$, $\Psi _{\theta }(P)$ provides the current optimal CCPs of an agent whose future behavior is according to $P$. 

Under the current assumptions, the policy function mapping has several properties that are central to the results of this paper.

\begin{lemma}\label{lem:PolicyProperties}
	Under Assumption \ref{ass:RegularityModel}, $ \Psi _{\theta }$ satisfies the following properties:
\begin{enumerate}[(a)]
\item $\Psi _{\theta }$ has a unique fixed point $P_{\theta }$,
\item The sequence $P^{K}=\Psi _{\theta } (P^{K-1})$ for $K\geq 1$, converges to $ P_{\theta }$ for any initial $P^{0} \in \Theta_P$,
\item The Jacobian matrix of $\Psi _{\theta }$ with respect to $P$ is zero at $P_{\theta }$.
\end{enumerate}
\end{lemma}

Following the literature on estimation of dynamic discrete choice models, the researcher estimates $\theta = (\alpha,\theta_{f})$ using a two-step procedure. In a first step, he uses $f_{\theta _{f}}$ to estimate $\theta _{f}$. In a second step, he uses $\Psi_{(\alpha,\theta _{f})}(P)$ and the first step to estimate $\alpha $. The following assumption ensures that the model is identified.

\begin{assumptionA} \label{ass:Identification}
	The parameter $\theta = (\alpha,\theta_{f}) \in \Theta $ is identified as follows:
\begin{enumerate}[(a)]
	\item $\theta _{f}$ is identified by $f_{\theta_{f}}$, i.e., $f_{\theta _{f,a}} =f_{\theta _{f,b}} $ implies $\theta _{f,a}=\theta _{f,b}$,
	\item $\alpha $ is identified by the fixed point condition $\Psi _{(\alpha ,\theta _{f})}(P)=P$ for any $(\theta_{f},P) \in \Theta_{f}\times \Theta_{P}$, i.e., $\forall \theta _{f}\in \Theta _{f}$, $\Psi _{(\alpha _{a},\theta _{f})}(P)=P$ and $\Psi _{(\alpha_{b} ,\theta _{f})}(P)=P$ implies $\alpha _{a}=\alpha _{b}$.
\end{enumerate}
\end{assumptionA}
\cite{magnac/thesmar:2002} provide sufficient conditions for Assumption \ref{ass:Identification}. Also, Assumption \ref{ass:Identification} implies the higher level condition used by \citet[conditions (e)-(f) in Proposition 4]{aguirregabiria/mira:2002}. Under these conditions, we can deduce certain important properties for the model-implied CCPs.

\begin{lemma} \label{lem:AuxResultsOnCCP}
Under Assumptions \ref{ass:RegularityModel}-\ref{ass:Identification},
\begin{enumerate}[(a)]
\item $P_{\theta }$ is continuously differentiable,
\item $\partial P_{\theta }/\partial \theta = \partial \Psi _{\theta }(P_{\theta })/\partial \theta$,
\item $\alpha $ is identified by $P_{(\alpha ,\theta _{f})}$ for any $\theta _{f}\in \Theta _{f}$, i.e., $\forall \theta _{f}\in \Theta _{f}$, $P_{(\alpha _{a},\theta _{f})}=P_{(\alpha _{b},\theta _{f})}$ implies $\alpha _{a}=\alpha _{b}$.
\end{enumerate}
\end{lemma}

Lemmas \ref{lem:PolicyProperties} and \ref{lem:AuxResultsOnCCP} are well-known results under correct specification. At the risk of being repetitive, we include these in the paper for two reasons. First, we note that these properties belong to the econometric model, regardless of whether it is correctly specified or not. Second, later results in the paper will repeatedly refer to these properties.

Thus far, we have described how the model specifies two conditional distributions: the CCPs and the transition probabilities. The final element of the specification is the marginal distribution of the state variables, which is left completely unspecified.

\subsection{Local misspecification}\label{sec:Misspecification}

We now describe the true data generating process (DGP), denoted by $\Pi^{*} _{n}(a,x,x')$, and explain its relationship to the econometric model in Section \ref{sec:Model}. Hereafter, a superscript with asterisk denotes true value.

By definition, the DGP is the product of the transition probability, the CCPs, and the marginal distribution of the state variable, i.e., for all $(a,x,x') \in A\times X\times X$,
\begin{equation}
\Pi_{n}^{*}(a,x,x' )~=~f_{n}^{\ast }(x'|a,x) ~\times ~P_{n}^{\ast }(a|x) ~\times ~m_{n}^{\ast }(x),
\label{eq:JointDistribution}
\end{equation}
where:
	\begin{align}
		{f}_{n}^{\ast}(x^{\prime }|a,x) ~&\equiv~  
			\frac{ \Pi_{n}^{\ast }( a,x,x^{\prime }) }{ \sum_{\tilde{x}^{\prime }\in X} \Pi_{n}^{\ast }( a,x,\tilde{x}^{\prime }) }\notag\\
		{P}_{n}^{*}(a|x) ~&\equiv~ 
			\frac{\sum_{\tilde{x}^{\prime }\in X} \Pi_{n}^{\ast }( a,x,\tilde{x} ^{\prime })}{\sum_{(\check{a},\check{x}')\in A \times X} \Pi_{n}^{\ast }( \check{a},x,\check{x}^{\prime })}\notag\\
		{m}_{n}^{*}(x) ~&\equiv~ 
			{\sum_{(a,x')\in A \times X} \Pi_{n}^{\ast }( {a},x,x')}. \label{eq:DefnPopulations}
	\end{align}
For the same reason as before, ${P}_{n}^{*}(|A||x) = 1- \sum_{a\in \tilde{A}} {P}_{n}^{*}(a|x)$ for all $x\in X$. Thus, the vector of true CCPs is completely characterized by $P_{n}^{*} \equiv \{{P}_{n}^{*}(a|x)\}_{(a,x)\in \tilde{A} \times X} \in \Theta_{P}$.

Section \ref{sec:Model} specifies $P_{\theta }(a|x)$ as the econometric model for $P_{n}^{\ast }(a|x)$ and $f_{\theta _{f}}(x'|a,x)$ as the econometric model for $f_{n}^{\ast }(x'|a,x)$. This paper allows the econometric model to be misspecified, i.e.,
\begin{equation}
\inf_{(\alpha ,\theta _{f})\in \Theta _{\alpha }\times \Theta _{f}}\Vert ~(P_{(\alpha ,\theta _{f})}-P_{n}^{\ast })'~,~(f_{\theta _{f}}-f_{n}^{\ast })'~\Vert > 0, \label{eq:Misspecification}
\end{equation}
but requires the misspecification to vanish asymptotically according to the following assumption.

\begin{assumptionA}\label{ass:LocalMiss} The model is locally misspecified in the following sense.
\begin{enumerate}[(a)]
	\item The sequence of DGPs $\{\Pi_{n}^{\ast }\}_{n\geq 1}$ with $\Pi_{n}^{\ast } \equiv \{\Pi_{n}^{\ast }(a,x,x')\}_{(a,x,x') \in A \times X \times X}$ converges to a limiting DGP $\Pi^{\ast } \equiv \{\Pi^{\ast }(a,x,x')\}_{(a,x,x') \in A \times X \times X}$ in the following manner:
	\begin{equation*}
	{n}^{\delta}(\Pi_{n}^{\ast } - \Pi^{\ast }) \to B_{\Pi^{*}} \in \mathbb{R}^{A \times X \times X},
	\end{equation*}
	where $\delta>0$ is an unknown parameter to the researcher.
 	\item The econometric model is correctly specified in the limit:
\begin{equation*}
\inf_{(\alpha ,\theta _{f})\in \Theta _{\alpha }\times \Theta _{f}}\Vert ~(P_{(\alpha ,\theta _{f})}-P^{\ast })'~,~(f_{\theta _{f}}-f^{\ast })'~ \Vert =0,
\end{equation*}
	where $P^{\ast } \equiv \lim_{n \to \infty} {P}_{n}^{*}$ denotes the limiting vector of CCPs and $f^{\ast }\equiv\lim_{n \to \infty} {f}_{n}^{\ast}$ denotes the limiting transition probabilities.
\end{enumerate}
\end{assumptionA}

Assumption \ref{ass:LocalMiss} describes the effect of the local misspecification on the distribution of the data. This high-level condition represents a situation in which the underlying structural features of the econometric model are ``close'' to the true ones. We now provide three illustrations. 

As our first illustration, suppose that the researcher incorrectly specifies one of the functional forms in the econometric model. For instance, he could specify that the utility function under action $a=1$ is a linear function of the state variable $x$ when, in reality, this function is quadratic:
\begin{equation*}
u_{\theta}(x,a=1) ~=~ \theta_1 + \theta_2 x + \tau_{n} x^2.
\end{equation*}
The coefficient of the quadratic term, $\tau_{n}$, controls the degree of misspecification. In particular, $\tau_{n} \to 0$ implies that, in the limit, the econometric model is correctly specified.

As a second illustration, suppose that the data are generated by the presence of unobserved heterogeneity along the lines of \cite{arcidiacono/miller:2011}.\footnote{We thank an anonymous referee for suggesting this second illustration.} Specifically, assume that the sample is composed of two types of agents, A and B. Both types of agents behave exactly according to the model but they differ in the parameter values of any of the structural functions (e.g.\ the utility function). If we use $\tau_{n} \in (0,1)$ to denote the proportion of agents of type B in the population, the observed true CCP $P^{*}_{n}$ is the mixture of CCPs of agents of type A and B with weights $(1-\tau_{n})$ and $\tau_{n}$, respectively. The model is then misspecified in the sense that it presumes a homogenous sample. As in the first illustration, $\tau_{n} \to 0$ implies that the model is correctly specified in the limit.

As a third illustration, suppose that the agent exhibits small departures from the predicted rational behavior according to the econometric model.\footnote{We thank the Editor, Chris Taber, for providing this third illustration.} Let us develop this idea in more detail. Given state variables $(x,\epsilon)$, our model predicts that the agent chooses the action $a$ that maximizes expected discounted utility, i.e.,
\begin{equation*}
	P_{n}^{\ast }(a|x,\epsilon)~=~1\left[a~=~\underset{\tilde{a}\in A}{\arg \max} 	\left( u_{\theta _{u}}( x,\tilde{a}) +\beta \sum\nolimits_{x'\in X}V_{\theta }( x') f_{\theta _{f}}( x'|x, \tilde{a}) +\epsilon(\tilde{a}) \right)\right].
	%\label{eq:ModelPrediction}
\end{equation*}
Instead, suppose that the agent chooses actions according to a multinomial distribution with choice probabilities that are increasing in the action-specific expected discounted utility. For example, given $(x,\epsilon)$, the agent chooses action $a \in A$ with probability:
\begin{equation}
	P_{n}^{\ast }(a|x,\epsilon)~=~\frac{\exp\left[\left( u_{\theta _{u}}( x,{a}) +\beta \sum\nolimits_{x'\in X}V_{\theta }( x') f_{\theta _{f}}( x'|x, {a}) +\epsilon({a}) \right)/\tau_{n}\right]}{\sum_{\tilde{a} \in A} \left[\exp \left( u_{\theta _{u}}( x,a) +\beta \sum\nolimits_{x'\in X}V_{\theta }( x') f_{\theta _{f}}( x'|x, a) +\epsilon(\tilde{a}) \right)/\tau_{n}\right]}.
	\label{eq:Irrational}
\end{equation}
The parameter $\tau_{n}\ge 0$ controls the degree of departure from rational behavior. Once again, note that $\tau_{n} \to 0 $ implies that the model is correctly specified in the limit. Finally, note that the CCPs $P_{n}^{\ast }(a|x)$ follow from integrating $\epsilon$ out from $P_{n}^{\ast }(a|x,\epsilon)$.

Despite being very different from a conceptual viewpoint, these three illustrations can all be framed in terms of Assumption \ref{ass:LocalMiss}. We now briefly explain this. First, these examples generate a discrepancy between the model-implied CCPs $P_{\theta }$ and the true CCPs $P_{n}^{\ast }$, i.e., $\left\Vert P_{n}^{\ast }-P_{\theta }\right\Vert >0$ for all $\theta \in \Theta$, i.e., Eq.\ \eqref{eq:Misspecification} follows. Second, in all cases, the parameter $\tau_n$ determines the amount model misspecification. If this parameter is close to zero, a continuity argument implies that the model-implied CCPs should be ``close'' to the true ones. In particular, if $\tau_n = O(n^{-\delta})$ for some $\delta>0$ and if the model CCPs are sufficiently smooth, it follows that:
\begin{enumerate}[(a)]
	\item $n^{\delta}(P_{n}^{\ast }-P^{\ast }) \to C \in \mathbb{R}^{{\tilde A} \times X}$,
	\item $P^{\ast } = P_{(\alpha^{*},\theta^{*}_f) }$ for some $(\alpha^{*},\theta^{*}_f) \in \Theta$.
\end{enumerate}
Assumption \ref{ass:LocalMiss} then follows from this and the correct specification of the transition probabilities. We note in passing that the first illustration is used as the framework for our Monte Carlo simulations.

% CHANGED BELOW
Assumption \ref{ass:LocalMiss}(a) requires that the local misspecification vanishes at a rate of $n^{\delta}$ for some $\delta>0$. This rate depends on the difference between the model-implied CCPs and the true CCPs and, thus, it is unknown to the researcher. Our framework allows this rate to be faster, equal, or even slower than the parametric rate $\sqrt{n}$. This is more general than the typical local misspecification framework that usually restricts to $\delta \geq 1/2$ (e.g., see \cite{newey:1985a,newey:1985b}, \cite{tauchen:1985}, \cite{bugni/canay/guggenberger:2012}).

Under these conditions, Theorem \ref{thm:LocalMisspecification} demonstrates that there is a unique true limiting parameter value $(\alpha ^{\ast },\theta _{f}^{\ast })$. As we later show, the estimators considered in this paper will converge in probability to this parameter value despite the (local) misspecification.

\begin{theorem}\label{thm:LocalMisspecification}
	Under Assumptions \ref{ass:Identification} and \ref{ass:LocalMiss}(b), there is a unique $(\alpha ^{\ast },\theta _{f}^{\ast })\in \Theta $ such that $P_{(\alpha ^{\ast },\theta _{f}^{\ast })}=P^{\ast }$ and $f_{\theta _{f}^{\ast }}=f^{\ast }$.
\end{theorem}

\section{General result for two-step $K$-stage estimators}\label{sec:Inference}

This paper considers two-step estimators based on the $K$-stage sequential policy function iteration (PI) algorithm developed by \cite{aguirregabiria/mira:2002}.\footnote{Results for single-step estimators are easy to derive from our analysis by considering the special case in which the entire parameter vector is estimated on the second step.} For any $K \in \mathbb{N}$, this estimator is defined as follows:
\begin{itemize}
	\item Stage 1: Estimate $\theta^*_{f}$ with a first-step estimator, denoted by $\hat{\theta} _{f_{n}}$. Also, estimate $P^*$ with the initial or $0$-stage estimator of the CCPs, denoted by $\hat{P}_{n}^{0}$.
	\item Stage 2: Estimate $\alpha^{*}$ with $\hat{\alpha}_{n}^{K}$, computed using the following algorithm. Initialize $k=1$ and then:
	\begin{itemize}
		\item[(a)] Compute:
		\begin{equation}
\hat{\alpha}_{n}^{k} ~\equiv~\underset{{\alpha \in \Theta _{\alpha }}}{\arg \max}~{Q}_{n}( \alpha ,
\hat{\theta}_{f_{n}}, \hat{P}_{n}^{k-1} ),\label{eq:k-StepDefn}
\end{equation}
where ${Q}_{n}: \Theta_{\alpha} \times \Theta_{f} \times \Theta_{P} \to \mathbb{R} $ is the sample objective function. If $k=K$, exit the algorithm. If $k<K$, go to (b).
	\item[(b)] Estimate $P^*$ with the $k$-stage estimator of the CCPs, given by:
		\begin{equation*}
			\hat{P}_{n}^{k}~\equiv~ \Psi _{(\hat{\alpha}_{n}^{k},\hat{\theta}_{f_{n}})}(\hat{P}_{n}^{k-1}).
		\end{equation*}
Then, increase $k$ by one unit and return to (a).  
	\end{itemize}
\end{itemize}

For any $K \in \mathbb{N}$, we estimate $\theta^* = (\theta^*_{f}, \alpha^*)$ with $\hat\theta_{n}^{K} \equiv (\hat{\theta} _{f_{n}},\hat{\alpha}_{n}^{K})$. This algorithm leaves several aspects of the estimation method unspecified: the estimators $\hat{\theta} _{f_{n}}$ and $\hat{P}_{n}^{0}$ and the sample criterion function ${Q}_{n}$. Our strategy in this section is to produce a general result without specifying these objects and based on high-level conditions. Section \ref{sec:Applications} will apply the general result to concrete estimators used in practice.  

\begin{assumptionA}\label{ass:EEpart1}	
	$\alpha ^{\ast }$ belongs to the interior of $\Theta _{\alpha }$.
\end{assumptionA}

\begin{assumptionA}\label{ass:EEpart2}	
	Let $\mathcal{N}$ denote an arbitrary small neighborhood of $( \alpha ,\theta _{f},P) $ around $( \alpha ^{\ast },\theta _{f}^{\ast },P^{\ast }) $. Then, there is a limiting function $Q_{\infty }: \Theta_{\alpha} \times \Theta_{f} \times \Theta_{P} \to \mathbb{R} $ such that:
	\begin{enumerate}[(a)]
	\item $\sup_{\alpha \in \Theta _{\alpha }}\vert {Q}_{n}( \alpha ,\hat{\theta}_{f_{n}},\tilde{P}_{n}) -Q_{\infty }( \alpha ,\theta _{f}^{\ast },P^{\ast }) \vert =o_{p_{n}}(1)$, provided that $\tilde{P}_{n} = P^{\ast } + o_{p_{n}}(1)$.
	\item $Q_{\infty }(\alpha ,\theta _{f}^{\ast },P^{\ast })$ is uniquely maximized at $\alpha ^{\ast }$.
	\item For any $\lambda \in \{\alpha ,\theta _{f},P\}$, $\partial^2 {Q}_{n}( \alpha ,\theta _{f},P) /\partial \alpha \partial \lambda'$ is a continuous function for all $( \alpha ,\theta _{f},P) \in \mathcal{N}$ w.p.a.1.
	\item For any $\lambda \in \{\alpha ,\theta _{f},P\}$, $\sup_{( \alpha ,\theta _{f},P) \in \mathcal{N}}\Vert \partial^2 {Q}_{n}( \alpha ,\theta _{f},P) /\partial \alpha \partial \lambda' -\partial^2 {Q}_{\infty }( \alpha ,\theta _{f},P) /\partial \alpha \partial \lambda' \Vert =o_{p_{n}}(1)$.
	\item $\partial^2 {Q}_{\infty }( \alpha ,\theta _{f},P)/\partial \alpha \partial \alpha ^{\prime }$ is a continuous function and non-singular
	at $(\alpha ^{\ast },\theta _{f}^{\ast },P^{\ast })$.
	\item For any $\lambda \in \{\alpha ,\theta _{f}\}$, $\partial^2 {Q}_{\infty }(\alpha ^{\ast },\theta _{f}^{\ast },P^{\ast })/\partial \lambda \partial P' =\mathbf{0}_{ d_{\lambda } \times |\tilde{A} \times X|}$.
	\end{enumerate}
\end{assumptionA}

\begin{assumptionA}\label{ass:EEpart3}	The following results hold:
\begin{enumerate}[(a)]
	\item ${n}^{\min\{1/2,\delta\}}[ \partial {Q}_{n}(\alpha ^{\ast },\theta _{f}^{\ast },P^{\ast })/\partial \alpha ^{\prime },( \hat{ \theta}_{f,n}-\theta _{f}^{\ast })']' \overset{d}{\to } \zeta  = [\zeta_1',\zeta_2']'$, for some random variable $\zeta $.
	\item ${n}^{\min\{1/2,\delta\}}(\hat{P}^{0}_{n} - P^{\ast }) = O_{p_{n}}(1)$.
\end{enumerate}
\end{assumptionA}

The second step of the $K$-stage PI algorithm is an iterative version of an extremum estimator. Assumption \ref{ass:EEpart1} is a standard assumption for extremum estimators, which allows us to rely on the first order conditions of the optimization in Eq.\ \eqref{eq:k-StepDefn}. With the exception of Assumption \ref{ass:EEpart2}(f), Assumption \ref{ass:EEpart2} is composed of the usual regularity conditions for extremum estimators under a drifting sequence of DGPs. Assumption \ref{ass:EEpart2}(f) is critical to establish the main result in this section and we will show that it is a consequence of the zero Jacobian property proved in Lemma \ref{lem:PolicyProperties}(d). Assumption \ref{ass:EEpart3} requires that certain random variables converge in distribution or are bounded in probability. The rate of convergence for these variables is ${n}^{\min\{1/2,\delta\}}$, i.e., the slowest rate between the local misspecification and the regular rate for parametric estimation. Assumption \ref{ass:EEpart3}(a) does not specify the distribution of $\zeta $, as this is not required to establish the general result in this section. Section \ref{sec:Applications} will apply the general result and will verify that $\zeta $ is a multivariate normal random variable with possibly non-zero mean. Finally, note that Assumptions \ref{ass:EEpart2}(f), (d), and \ref{ass:EEpart3}(b) use the subscript ${p_{n}}$ to refer to a drifting sequence of DGPs. This is necessary in our paper to handle the presence of the local misspecification.

Under these assumptions, Theorem \ref{thm:ANproof} establishes the asymptotic distribution of the two-step $K$-stage policy function iteration estimator.

\begin{theorem}[{\bf General result}]\label{thm:ANproof} 
Assume Assumptions \ref{ass:RegularityModel}-\ref{ass:EEpart3}. For any $K\geq 1$,
\begin{align*}
&{n}^{\min\{1/2,\delta\}}(\hat{\alpha}_{n}^{K}-\alpha ^{\ast })\\
&= -\left(\frac{\partial^{2} {Q} _{\infty }(\alpha ^{\ast },\theta _{f}^{\ast },P^{\ast })}{\partial \alpha \partial \alpha ^{\prime }}\right)^{-1}{n}^{\min\{1/2,\delta\}} \left[\frac{\partial {Q} _{n}(\alpha ^{\ast },\theta _{f}^{\ast },P^{\ast })}{\partial \alpha }+\frac{\partial^{2} {Q}_{\infty }(\alpha ^{\ast },\theta _{f}^{\ast },P^{\ast })}{\partial \alpha \partial \theta _{f}'}(\hat{\theta}_{f,n}-\theta _{f}^{\ast })\right] + o_{p_{n}}(1)\\
&\overset{d}{\to } -\left(\frac{\partial^{2} {Q}_{\infty }(\alpha ^{\ast },\theta _{f}^{\ast },P^{\ast })}{\partial \alpha \partial \alpha ^{\prime }}\right)^{-1}\left[\zeta_{1}+\frac{\partial^{2} {Q}_{\infty }(\alpha ^{\ast },\theta _{f}^{\ast },P^{\ast })}{\partial \alpha\partial \theta _{f}'}\zeta_{2}\right].
\end{align*}
\end{theorem}

The result reveals two important features of the asymptotic distribution of our estimators. First, the local misspecification vanishing at the rate of $n^{\delta}$ causes the estimator to converge to the true limiting structural parameter at a rate of ${n}^{\min\{1/2,\delta\}}$. Second, the asymptotic distribution is invariant to the number of iterations $K$. In fact, the first equality in Theorem \ref{thm:ANproof} implies that changes in $K$ are asymptotically irrelevant. This result holds regardless of the rate of local misspecification $\delta$. The invariance of the asymptotic distribution to $K$ is one of the main findings of this paper. The intuition of this result is as follows. By evaluating Eq.\ \eqref{eq:k-StepDefn} at $k=K$, we find that:
\begin{equation*}
	\hat{\alpha}_{n}^{K} ~~\equiv~~\underset{{\alpha \in \Theta _{\alpha }}}{\arg \max}~{Q}_{n}( \alpha ,
\hat{\theta}_{f_{n}}, \hat{P}_{n}^{K-1} ).
\end{equation*}
This equation reveals that local misspecification can affect the asymptotic distribution of $\hat{\alpha}_{n}^{K}$ through three channels: the sample criterion function $Q_{n}$, the first-step estimator $\hat{ \theta}_{f,n}$, and the $(K-1)$-stage estimator of the CCPs $\hat{P}^{K-1}_{n}$. As the notation shows, the third channel depends explicitly on $K$, and its effect could potentially change with every iteration. This is not the case, and we now explain why. Our formal argument shows that the zero Jacobian property (embodied in Assumption \ref{ass:EEpart2}(f)) effectively erases the accumulated effect of the model misspecification in all previous iterations. Then, model misspecification can only affect the estimator through its most recent iteration. Our formal arguments also show that the effect of the most recent iteration is invariant with the number of iterations. From these two observations, the result follows.

\section{Applications of the general result}\label{sec:Applications}

We now apply Theorem \ref{thm:ANproof} to classes of estimators used in practice. Section \ref{sec:ML} considers $K$-ML estimation and Section \ref{sec:MD} considers $K$-MD estimation. 

Throughout this section, we presume that the researcher observes an i.i.d.\ sample distributed according to the true (drifting) DGP.

\begin{assumptionA}\label{ass:iid}
	For each $n \in \mathbb{N}$, $\{(a_{i},x_{i},x_{i}^{\prime })\}_{i\leq n}$ is an i.i.d.\ sample distributed according to $\Pi^{*} _{n}(a,x,x')$.
\end{assumptionA}

Under this assumption, it is natural to consider the sample analogue estimators of the DGP, the CCPs, and the transition probabilities, i.e., for all $(a,x,x') \in A \times X \times X$,
		\begin{align}
			\hat{\Pi}_{n}( a,x,x^{\prime })~&\equiv~  \sum_{i=1}^{n}1[ x_{i}=x,a_{i}=a,x_{i}^{\prime }=x^{\prime }] /n \notag\\
		\hat{P}_{n}(a|x) ~&\equiv~ 
		 \frac{\sum_{\tilde{x}^{\prime }\in X} \hat{\Pi}_{n}( a,x,\tilde{x} ^{\prime })}{\sum_{(\check{a},\check{x}')\in A \times X} \hat{\Pi}_{n}( \check{a},x,\check{x}^{\prime })}\notag\\
			\hat{f}_{n}(x^{\prime }|a,x) ~&\equiv~  
			\frac{ \hat{\Pi}_{n}( a,x,x^{\prime }) }{ \sum_{\tilde{x}^{\prime }\in X} \hat{\Pi}_{n}( a,x,\tilde{x}^{\prime }) }.\label{eq:DefnSamples}
		\end{align}

We now consider a general framework for the preliminary estimators in the algorithm.

\begin{assumptionA}\label{ass:Preliminary}
For $\hat{\Pi}_{n} \equiv \{\hat{\Pi}_{n}(a,x,x')\}_{(a,x,x') \in A \times X \times X}$, assume that:
	\begin{equation*}
	(\hat{\theta}_{f,n}, \hat{P}^{0}_{n})~=~G( \hat{\Pi}_{n} ),
	\end{equation*}
	where the function $G: \mathbb{R} ^{|A \times X \times X|} \to \mathbb{R} ^{d_{\theta_{f}}} \times \Theta_P $ is continuously differentiable at $\Pi^{*}$ and $(\theta _{f}^{\ast },P^{\ast })=G(\Pi^{*}) $.
\end{assumptionA}

Assumption \ref{ass:Preliminary} is very mild. By the identification result in Theorem \ref{thm:LocalMisspecification} and Assumption \ref{ass:iid}, it is reasonable to presume that the researcher estimates $(\theta_f^*,P^*)$ using a smooth function of the sample analogue estimator of the DGP.
In fact, Assumption \ref{ass:Preliminary} is automatically satisfied if we use a sample analogue estimator of the CCPs, i.e., $\hat{P}^{0}_{n}=\hat{P}_{n}$, and a non-parametric model for the transition probability that is estimated by sample analogues, i.e., $\theta _{f}\equiv \{f(x^{\prime }|a,x)\}_{(a,x,x^{\prime })\in A \times X \times X}$ and $\hat{\theta}_{f,n}=\hat{f}_{n}$.\footnote{One practical problem with the sample analogue estimators is that they are undefined if ${ \sum_{x'\in X}\hat{\Pi}_{n}( a,x,x') }=0$ for any $(a,x)\in A \times X$. For a discussion of this, see \cite{hotz/miller/sanders/smith:1994} and \citet[Page 914]{pesendorfer/schmidt-dengler:2008}. Of course, this is only a problem in small samples and does not affect the validity of our asymptotic arguments.}

\subsection{$K$-ML estimators}\label{sec:ML}

We now specialize the general result to $K$-ML estimators considered by \cite{aguirregabiria/mira:2002}. This is achieved by setting the sample objective function ${Q}_{n}$ to the pseudo-likelihood function, i.e.,
\begin{eqnarray*}
{Q}_{n}^{ML}( \alpha ,\theta _{f},P) ~\equiv ~n^{-1}\sum_{i=1}^{n}\ln \Psi _{(\alpha ,\theta _{f})}( P) ( a_{i}|x_{i}).
\end{eqnarray*}

To derive the asymptotic distribution of the $K$-ML estimator, we impose the following regularity conditions. 

\begin{assumptionA}\label{ass:Regularity}$\Psi $ satisfies the following properties:
\begin{enumerate}[(a)]
	\item $ \Psi _{\theta }(P)(a|x) \in (0,1)$ for any $(a,x) \in \tilde{A} \times X$.
	\item $ \Psi _{\theta }(P)$ is twice continuously differentiable in $\theta$ and $P$.
	\item $\partial\Psi _{\theta }(P)/\partial \alpha'$ is a full rank matrix at $(\theta^{*},P_{\theta^{*}})$.
\end{enumerate}	
\end{assumptionA}

Assumption \ref{ass:Regularity} are connected with the requirements in Assumption \ref{ass:EEpart2} and are standard in the literature. Assumption \ref{ass:Regularity}(a)-(b) are identical to \citet[conditions (b)-(c) of Proposition 4]{aguirregabiria/mira:2002}. Assumption \ref{ass:Regularity}(c) is connected to the non-singularity requirement in Assumption \ref{ass:EEpart2}(e). Since the $\alpha$ has been assumed to be identified by $\Psi _{\theta }(P)=P$ (and, thus, locally identified by it), Assumption \ref{ass:Regularity}(c) is equivalent to the regularity conditions in \citet[Theorem 1]{rothenberg:1971}.

Theorem \ref{thm:ML_main} is a corollary of Theorem \ref{thm:ANproof} and characterizes the asymptotic distribution of the $K$-ML estimator under local misspecification.

\begin{theorem}[{\bf $K$-ML}] \label{thm:ML_main}
Assume Assumptions \ref{ass:RegularityModel}-\ref{ass:EEpart1} and \ref{ass:iid}-\ref{ass:Regularity}. Then, for any $K,\tilde{K}\geq 1$,
\begin{align*}
{n}^{\min\{1/2,\delta\}} (\hat{\alpha}_{n}^{K-ML}-\alpha ^{\ast }) &={n}^{\min\{1/2,\delta\}}(\hat{\alpha}_{n}^{ \tilde{K}-ML}-\alpha ^{\ast })+o_{p_{n}}(1) \\
&\overset{d}{\to } \Upsilon _{ML} \times  \Delta \times  N\left( ~ B_{\Pi^{*}}\times 1[\delta \leq 1/2]~,~ 
~ (diag(\Pi^{*}) - \Pi^{*} {\Pi^{*}}' )\times 1[\delta \geq 1/2]~
\right)
\end{align*}
where $B_{\Pi^{*}}$ and $\Pi^{*}$ are as in Assumption \ref{ass:LocalMiss}, and $\Upsilon _{ML}$ and $\Delta$ are the following matrices:
\begin{align}
\Upsilon _{ML} ~ &\equiv ~\left(\frac{\partial P_{\theta ^{\ast }}'}{\partial \alpha } \Phi \frac{\partial P_{\theta ^{\ast }} }{\partial \alpha' } \right)^{-1}\frac{\partial P_{\theta ^{\ast }}'}{ \partial \alpha } \Phi  \left[
	\begin{array}{cc}
	\Sigma & -\dfrac{\partial P_{\theta ^{\ast }}}{\partial \theta _{f}'}
	\end{array}
	\right] \in \mathbb{R}^{d_{\alpha} \times (| A \times X| + d_{\theta_{f}})} \notag\\
\Delta ~&\equiv~
	\left[ 
	\begin{array}{cccc}
	{\bf I}_{|A \times X| \times |A \times X|}&{\bf I}_{|A \times X| \times |A \times X|}&\dots & {\bf I}_{|A \times X| \times |A \times X|} \\
	&\dfrac {\partial G_{1}(\Pi^{*})}{\partial {\Pi^{*}}'}&&
	\end{array}
	\right] \in \mathbb{R}^{(| A \times X| + d_{\theta_{f}}) \times | A \times X \times X|}, \label{eq:Delta}
\end{align}
where $G_{1}$ denotes the first component of $G$ in Assumption \ref{ass:Preliminary}, i.e., $\hat{\theta}_{f,n} \equiv G_{1}(\hat{\Pi}_{n})$,
\begin{eqnarray}
\Phi ~&\equiv &~\left[ 
\begin{array}{cccc}
\Phi _{1} & \mathbf{0}_{|\tilde{A}|\times |\tilde{A}|} & \ldots & \mathbf{0}_{|\tilde{A}|\times | \tilde{A}|} \\
\mathbf{0}_{|\tilde{A}|\times |\tilde{A}|} & \Phi _{2} & \ldots  & \vdots \\
\vdots &  \ddots& \ddots & \mathbf{0}_{|\tilde{A}|\times |\tilde{A}|} \\
\mathbf{0}_{|\tilde{A}|\times |\tilde{A}|} & \ldots  & \mathbf{0}_{|\tilde{A}|\times |\tilde{A}|} & \Phi _{|X|}
\end{array}
\right]  \in \mathbb{R}^{|\tilde A \times X| \times |\tilde A \times X|} \notag \\
\Sigma ~&\equiv &~\left[ 
\begin{array}{cccc}
\Sigma _{1} & \mathbf{0}_{|\tilde{A}|\times |A|} & \ldots & \mathbf{0}_{|\tilde{A}|\times |A|} \\
\mathbf{0}_{|\tilde{A}|\times |A|} & \Sigma _{2} & \ldots  & \vdots \\
\vdots &  \ddots& \ddots & \mathbf{0}_{|\tilde{A}|\times |A|} \\
\mathbf{0}_{|\tilde{A}|\times |A|} & \ldots  & \mathbf{0}_{|\tilde{A}|\times |A|} & \Sigma _{|X|}
\end{array}
\right] \in \mathbb{R}^{|\tilde A \times X| \times | A \times X|}, \label{eq:Sigma}
\end{eqnarray}
and, finally, for all $x \in X$,
	\begin{align*}
		\Phi_{x} ~&\equiv~ m^{*}(x) \left[ diag \left\{ \{{1}/{P^{\ast }(a|x)}\}_{a\in \tilde{A}} \right\}  + \mathbf{1}_{|\tilde{A}|\times |\tilde{A}|} /({1-\sum\nolimits_{a \in \tilde{A}} P^{\ast }(a|x) })  \right] \in \mathbb{R}^{|\tilde A| \times | \tilde A |} \\
	 \Sigma _{x}~&\equiv~ \left[ \mathbf{I}_{|\tilde{A} |\times |A|}-\{P^{\ast }(a|x)\}_{a\in \tilde{A}} \times \mathbf{1} _{1\times |A|} \right] / {m^{*}(x)}\in \mathbb{R}^{|\tilde A| \times |  A |} \\
		m^{*}(x) ~&\equiv~ \sum_{(a,x') \in A \times X }{\Pi^{\ast }(a,x,x')} \in \mathbb{R}.
	\end{align*}
\end{theorem}

As discussed in Theorem \ref{thm:ANproof}, Theorem \ref{thm:ML_main} shows that changes in $K$ have an asymptotically negligible effect on the $K$-MD estimator. Theorem \ref{thm:ML_main} also shows that the rate of convergence of the estimator is ${n}^{\min\{1/2,\delta\}}$ and that its asymptotic distribution is normal with asymptotic bias and variance given by:
\begin{align*}
AB_{ML} &=\Upsilon _{ML} \times \Delta \times B_{\Pi^{*}} \times  1[\delta \leq 1/2],\\
AV_{ML} &=\Upsilon _{ML}\times  \Delta \times (diag(\Pi^{*}) - \Pi^{*} {\Pi^{*}}' ) \times \Delta' \times \Upsilon _{ML}'\times 1[\delta \geq 1/2].
\end{align*}

In the case of $\delta>1/2$, the local misspecification is irrelevant relative to sampling error and has no effect on the rate of convergence or the asymptotic distribution. A very different situation occurs when $\delta<1/2$. In this case, the local misspecification is overwhelming relative to sampling error and dominates the asymptotic distribution. The rate of convergence of the estimator is ${n}^{\delta}$ and, at this rate, the asymptotic distribution is a pure bias term. Finally, we have a knife-edge case with $\delta =1/2$. The rate of convergence is the usual parametric rate $\sqrt{n}$, but the asymptotic distribution can be biased. In consequence, an adequate characterization of the estimator's precision is the asymptotic mean squared error:
\begin{equation*}
\Upsilon _{ML}\times  \Delta \times (diag(\Pi^{*}) - \Pi^{*} {\Pi^{*}}' + B_{\Pi^{*}} B_{\Pi^{*}}'  ) \times \Delta' \times \Upsilon _{ML}'.
\end{equation*}

% We conclude the subsection with a comment regarding the asymptotic optimality of the $K$-ML estimator.
The $K$-ML estimator is a ``partial'' ML estimator in the sense that it plugs in the first-step estimator into the second step. In other words, it is not a ``full'' maximum likelihood estimator with respect to the entire parameter vector $\theta= (\alpha,\theta_f)$. Because of this feature, the usual optimality results for maximum likelihood estimation need not apply. In fact, the next subsection will describe a $K$-MD estimator that can be more efficient than the $K$-ML estimator, even in the absence of local misspecification.

\begin{remark}
	Theorem \ref{thm:ML_main} applies to any $K \in \mathbb{N}$ but does not extend to $K \to \infty$. Under some additional conditions, however, \citet[Proposition 3]{aguirregabiria/mira:2002} shows that if the $K$-ML estimator converges as $K \to \infty$, it will do so to a solution of \cite{rust:1987}'s nested fixed point estimator. As noted in \citet[Footnote 16]{aguirregabiria/mira:2002}, this result presumes the convergence of the $K$-ML estimator as $K \to \infty$, which has not been shown in the literature.
\end{remark}

\subsection{$K$-MD estimators}\label{sec:MD}

To specialize the general result to $K$-MD estimators, we set the sample objective function ${Q}_{n}$ to:
\begin{equation}
{Q}_{n}^{MD}(\alpha ,\theta _{f},P)~\equiv ~-[\hat{P}_{n}-\Psi _{(\alpha ,\theta _{f})}(P)]^{\prime }~\hat{W}_{n}~[\hat{P}_{n}-\Psi _{(\alpha ,\theta _{f})}(P)]. \label{eq:MDcriterion}
\end{equation}
where $\hat{P}_n$ is the sample frequency estimator of the vector of CCPs in Eq.\ \eqref{eq:DefnSamples} and $\hat{W}_{n} \in \mathbb{R}^{|\tilde{A} \times X| \times |\tilde{A} \times X|}$ is the weight matrix. In the special case of $K=1$ and $\hat{P}_{n}^{0}=\hat{P}_{n}$, the $K$-MD estimator coincides with the estimators considered in \cite{hotz/miller:1993} and \cite{pesendorfer/schmidt-dengler:2008}.\footnote{To be precise, \citet[Eqs.\ (18)-(19)]{pesendorfer/schmidt-dengler:2008} consider a sample criterion function that allows $\hat{P}_{n}$ in Eq.\ \eqref{eq:MDcriterion} to differ from the sample frequency estimator. We could also incorporate this feature in our setup at the expense of using longer arguments.} We impose the following condition regarding the weight matrix.

\begin{assumptionA} \label{ass:MD_WeightMatrix}
$\hat{W}_{n}=W^{\ast }+o_{p_{n}}(1)$, where $ W^{\ast } \in \mathbb{R}^{|\tilde{A} \times X| \times |\tilde{A} \times X|}$ is positive definite and symmetric.
\end{assumptionA}

In principle, we could generalize Assumption \ref{ass:MD_WeightMatrix} by allowing the weight matrix to be a function of the parameters of the problem. Similar results would then follow from longer arguments.

Theorem \ref{thm:MD_main} is a corollary of Theorem \ref{thm:ANproof} and characterizes the asymptotic distribution of the $K$-MD estimator under local misspecification.
\begin{theorem}[{\bf $K$-MD}] \label{thm:MD_main}
Assume Assumptions \ref{ass:RegularityModel}-\ref{ass:EEpart1} and \ref{ass:iid}-\ref{ass:MD_WeightMatrix}. Then, for any $K,\tilde{K}\geq 1$,
\begin{align*}
&{n}^{\min\{1/2,\delta\}} (\hat{\alpha}_{n}^{K-MD}-\alpha ^{\ast }) ={n}^{\min\{1/2,\delta\}}(\hat{\alpha}_{n}^{ \tilde{K}-MD}-\alpha ^{\ast })+o_{p_{n}}(1) \\
&\overset{d}{\to } \Upsilon _{MD}(W^*) \times  \Delta \times  N\left( ~ B_{\Pi^{*}}\times 1[\delta \leq 1/2]~,~ 
~ (diag(\Pi^{*}) - \Pi^{*} {\Pi^{*}}' )\times 1[\delta \geq 1/2]~
\right)
\end{align*}
where $B_{\Pi^{*}}$ and $\Pi^{*}$ are as in Assumption \ref{ass:LocalMiss}, $\Upsilon _{MD}(W^*)$ is the following matrix:
\begin{eqnarray*}
\Upsilon _{MD}(W^*) ~ &\equiv& ~\left(\frac{\partial P_{\theta ^{\ast }}'}{\partial \alpha} W^* \frac{\partial P_{\theta ^{\ast }} }{\partial \alpha' }\right)^{-1}\frac{\partial P_{\theta ^{\ast }}'}{ \partial \alpha } W^*  \left[
\begin{array}{cc}
\Sigma & -\dfrac{\partial P_{\theta ^{\ast }}}{\partial \theta _{f}'}
\end{array}
\right] \in \mathbb{R}^{d_{\alpha} \times (| A \times X| + d_{\theta_{f}})},
\end{eqnarray*}
with $\Delta$ is as in Eq.\ \eqref{eq:Delta} and $\Sigma$ is as in Eq.\ \eqref{eq:Sigma}.
\end{theorem}

\begin{remark}\label{rem:MLENonoptimal}
The asymptotic distribution of the $K$-ML estimator is a special case of that of the $K$-MD estimator with $W^{\ast }=\Phi$.
\end{remark}

As discussed in previous theorems, Theorem \ref{thm:MD_main} reveals that changes in $K$ have an asymptotically negligible effect on the $K$-MD estimator. This result also shows that the rate of convergence of the estimator is ${n}^{\min\{1/2,\delta\}}$ and that its asymptotic distribution is normal with asymptotic bias and variance that are given by:
\begin{align}
AB_{MD}(W^*) &=\Upsilon _{MD}(W^*) \times \Delta \times B_{\Pi^{*}} \times  1[\delta \leq 1/2], \notag\\
AV_{MD}(W^*) &=\Upsilon _{MD}(W^*)\times  \Delta \times (diag(\Pi^{*}) - \Pi^{*} {\Pi^{*}}' ) \times \Delta' \times \Upsilon _{MD}(W^*)'\times 1[\delta \geq 1/2].\label{eq:BiasVarMD}
\end{align}
In the knife-edge case with $\delta =1/2$, the asymptotic variance and bias can coexist. In consequence, an adequate characterization of the estimator's precision is the asymptotic mean squared error:
\begin{equation}
\Upsilon _{MD}(W^*)\times  \Delta \times (diag(\Pi^{*}) - \Pi^{*} {\Pi^{*}}' + B_{\Pi^{*}} B_{\Pi^{*}}'  ) \times \Delta' \times \Upsilon _{MD}(W^*)'.\label{eq:AMSEdefn}
\end{equation}

We now briefly discuss the optimality in the choice of $W^{\ast }$ in $K$-MD estimation. First, consider the case in which local misspecification is asymptotically irrelevant, i.e., $ \delta >1/2$. In this case, the $K$-MD estimator presents no asymptotic bias and the asymptotic variance and mean squared error coincide. Provided that relevant matrices are non-singular, standard arguments in GMM estimation (e.g.\ \citealp[Page 2165]{mcfadden/newey:1994}) imply that the minimum asymptotic variance and mean squared error among $K$-MD estimators are both equal to:
\begin{equation*}
% AV_{MD}^{\ast } = AMSE_{MD}^{\ast }~\equiv~
\left( \frac{\partial P_{\theta ^{\ast }}}{\partial \alpha' }'
\left[  \left[
\begin{array}{cc}
\Sigma & -\dfrac{\partial P_{\theta ^{\ast }}}{\partial \theta _{f}'}
\end{array}
\right] \Delta  (diag(\Pi^{*}) - \Pi^{*} {\Pi^{*}}'  )  \Delta' \left[
\begin{array}{cc}
\Sigma & -\dfrac{\partial P_{\theta ^{\ast }}}{\partial \theta _{f}'}
\end{array}
\right]' \right] ^{-1}
\frac{\partial P_{\theta ^{\ast }}}{\partial \alpha' } 
\right) ^{-1},
\end{equation*}
This minimum can be achieved by the following ``feasible'' choice of limiting weight matrix:
\begin{equation}
W_{AV}^{\ast } ~\equiv~ 
\left[  \left[
\begin{array}{cc}
\Sigma & -\dfrac{\partial P_{\theta ^{\ast }}}{\partial \theta _{f}'}
\end{array}
\right] \Delta  (diag(\Pi^{*}) - \Pi^{*} {\Pi^{*}}'  )  \Delta'  \left[
\begin{array}{cc}
\Sigma & -\dfrac{\partial P_{\theta ^{\ast }}}{\partial \theta _{f}'}
\end{array}
\right]' \right] ^{-1}
.\label{eq:W_AV}
\end{equation}
We say that Eq.\ \eqref{eq:W_AV} is a feasible choice because it can be consistently estimated. As pointed out in Remark \ref{rem:MLENonoptimal}, the $K$-ML estimator has the same asymptotic distribution as the $K$-MD estimator with $W_{ML}^{\ast }=\Phi$. Except under special conditions on the econometric model (e.g.\ $\partial P_{\theta^{\ast }}/\partial \theta _{f}'=\mathbf{0}_{|\tilde{A} \times X|\times d_{\theta_{f}}}$), the $K$-ML estimator is not necessarily optimal among the $K$-MD estimators.

Next, consider the knife-edge case in which local misspecification vanishes at the rate of sampling error, i.e., $\delta =1/2$. Once again, standard arguments in GMM estimation imply that the minimum asymptotic mean squared error among all $K$-MD estimators is:
\begin{equation*}
% AMSE_{MD}^{\ast } ~\equiv~
\left( \frac{\partial P_{\theta ^{\ast }}'}{\partial \alpha }
\left[  \left[
\begin{array}{cc}
\Sigma & -\dfrac{\partial P_{\theta ^{\ast }}}{\partial \theta _{f}'}
\end{array}
\right] \Delta  (diag(\Pi^{*}) - \Pi^{*} {\Pi^{*}}'  + B_{\Pi^{*}} B_{\Pi^{*}}'  )  \Delta'  \left[
\begin{array}{cc}
\Sigma & -\dfrac{\partial P_{\theta ^{\ast }}}{\partial \theta _{f}'}
\end{array}
\right]' \right] ^{-1}
\frac{\partial P_{\theta ^{\ast }}}{\partial \alpha' } 
\right) ^{-1}.
\end{equation*}
According to \citet[Page 2165]{mcfadden/newey:1994}, any limiting weight matrix $W^{*}$ that minimizes the asymptotic mean squared error should satisfy the following condition. For some matrix $C \in \mathbb{R}^{|\tilde{A}\times X| \times |\tilde{A}\times X|}$,
\begin{equation}
\frac{\partial P_{\theta ^{\ast }}}{\partial \alpha' } {W}^{\ast }= C \frac{\partial P_{\theta ^{\ast }}}{\partial \alpha' } \left[
\left[
\begin{array}{cc}
\Sigma & -\dfrac{\partial P_{\theta ^{\ast }}}{\partial \theta _{f}'}
\end{array}
\right] \Delta  (diag(\Pi^{*}) - \Pi^{*} {\Pi^{*}}'  + B_{\Pi^{*}} B_{\Pi^{*}}' )  \Delta'  \left[
\begin{array}{cc}
\Sigma & -\dfrac{\partial P_{\theta ^{\ast }}}{\partial \theta _{f}'}
\end{array}
\right]'
\right]^{-1}.\label{eq:W_AMSE_condition}
\end{equation}
In other words, Eq.\ \eqref{eq:W_AMSE_condition} characterizes the class of optimal limiting weight matrices. A simple example of an optimal limiting weight matrix is:
\begin{equation}
W_{AMSE}^{\ast }\equiv 
\left[  \left[
\begin{array}{cc}
\Sigma & -\dfrac{\partial P_{\theta ^{\ast }}}{\partial \theta _{f}'}
\end{array}
\right] \Delta  (diag(\Pi^{*}) - \Pi^{*} {\Pi^{*}}'  + B_{\Pi^{*}} B_{\Pi^{*}}' )  \Delta'  \left[
\begin{array}{cc}
\Sigma & -\dfrac{\partial P_{\theta ^{\ast }}}{\partial \theta _{f}'}
\end{array}
\right]' \right] ^{-1}.\label{eq:W_AMSE}
\end{equation}
By Eq.\ \eqref{eq:W_AMSE_condition}, the consistent estimation of any optimal limiting weight matrix requires the consistent estimation of the asymptotic bias. Unfortunately, the researcher is not aware of this feature of the population distribution. In this sense, estimating the optimal limiting weight matrix under local misspecification is ``infeasible'' in practice. In particular, note that the ``feasible'' limiting weight matrix $W_{AV}^{\ast }$ that minimizes asymptotic variance could fail to be optimal (i.e.\ Eq.\ \eqref{eq:W_AMSE_condition} is not satisfied).\footnote{This is the case in some of our Monte Carlo simulations, in which using $W_{AV}^{\ast }$ produces an asymptotic mean squared error that is larger than that the one obtained by using $\hat{W}_{n}={\bf I}_{|\tilde{A}\times X| \times |\tilde{A}\times X|}$.}

Finally, we could consider the case in which local misspecification is asymptotically overwhelming, i.e., $\delta <1/2$. In this case, the asymptotic distribution collapses to the pure bias term in Eq.\ \eqref{eq:BiasVarMD}, and the asymptotic mean squared error coincides with the square of the asymptotic bias. As in the case with $\delta =1/2$, minimizing asymptotic mean squared error is infeasible in the sense that it depends on the unknown asymptotic bias. Also, any ``feasible'' choice of limiting weight matrix such as $W_{AV}^{\ast }$ could fail to be optimal.

\section{Monte Carlo simulations}\label{sec:MonteCarlos} 

This section investigates the finite sample performance of the two-step estimators considered in previous sections under local misspecification. 

We simulate data using the classical bus engine replacement problem studied by \cite{rust:1987}. In each period $t =1,\dots,T\equiv \infty$, the bus owner has to decide whether to replace the bus engine or not to minimize the discounted present value of his costs. In any representative period, his choice is denoted by $a \in  A=\{1,2\}$, where $a=2$ represents replacing the engine, $a=1$ represents not replacing the engine, and the current engine mileage is denoted by $x \in X \equiv \{1,\dots,20\}$.

The researcher assumes that all individuals in his sample have the same deterministic part of the utility (profit) function, given by:
\begin{equation}
u_{\theta _{u}}(x,a)~=~-\theta _{u,1}\times  1[ a=2]~-\theta _{u,2}\times  1[ a=1]x,
\label{eq:utilityMC}
\end{equation}
where $\theta_{u} \equiv ( \theta_{u,1},\theta_{u,2}) \in \Theta_{u} \equiv [-B,B]^2$ with $B=10$. In addition, the researcher also assumes that the errors are distributed according to an extreme value type I distribution, independent of $x$, i.e.,
\begin{equation}
g(\epsilon =e|x)=\prod_{a\in A}\exp (e({a})) \exp ( -\exp ( e({a})) ),
\label{eq:errorMC}
\end{equation}
which does not have unknown parameters. Finally, the observed state is assumed to evolve according to the following Markov chain:
\begin{equation}
f_{\theta _{f}}(x'|x,a)~=~(1- \theta _{f}) \times  1[ a=1,x'=\min \{x+1,|X|\} ]~+~\theta _{f} \times  1[ a=1,x'=x] ~+~1[ a=2,x'=1],
\label{eq:transitionMC}
\end{equation}
where $\theta_{f} \in \Theta_{f} \equiv [0,1]$. The researcher correctly assumes that $\beta=0.9999$. His goal is to estimate $\theta =( \alpha ,\theta _{f}) \in \Theta =\Theta _{\alpha }\times \Theta _{f}$ with $\alpha = \theta_{u} \in \Theta _{\alpha } = \Theta_{u}$.

The researcher correctly specified the error distribution and state transition probabilities, which satisfy Eq.\ \eqref{eq:transitionMC} with $\theta_f = 0.25$. Unfortunately, he does not correctly specify the utility function. The correct utility function is as follows:
\begin{equation}
u_{\theta _{u},n}(x,a) ~=~ -\theta_{u,1} \times  1[a=2] ~ -\theta_{u,2} \times  1[a=1]x ~ +\tau_{n} \times  1[a=1]x^2,
\label{eq:utilityTrueMC}
\end{equation}
with $\theta _{u,1} = 1$, $\theta _{u,2} = 0.05$, and $\tau_{n} = -0.025 n^{-\delta}$ with $\delta \in \{1/3,1/2,1\}$. Notice that this form of model misspecification is analogous to the first illustration discussed in Section \ref{sec:Misspecification}.\footnote{Section \ref{sec:Misspecification} provides two other examples of local misspecification. We have also conducted Monte Carlo simulations based on these designs. For the sake of brevity, these are presented in the Supplement.} By the arguments in Section \ref{sec:Setup}, the true CCPs $P^{*}_{n}$ are determined by the true error distribution (Eq.\ \eqref{eq:errorMC}), true state transition probabilities (Eq.\ \eqref{eq:transitionMC} with $\theta _{f}=0.25$), and the true utility function (Eq.\ \eqref{eq:utilityTrueMC}). Our choices of $\delta$ include a case in which the local misspecification is asymptotically irrelevant ($\delta=1$), one case in which local misspecification is the knife-edge case ($\delta=1/2$), and one case in which the local misspecification is overwhelming ($\delta = 1/3$). In addition, we also consider a case in which the econometric model is correctly specified.

Our simulation results will be the average of $S=20,000$ independent datasets of observations $\{(a_{i},x_{i},x_{i}')\}_{i \leq n}$ that are i.i.d.\ distributed according to $\Pi_{n}^{*}$. We present simulation results for sample sizes of $n \in \{200,~500,~1,000\}$.
We generate marginal observations of the state variables according to the following distribution:
\begin{equation*}
	m^{*}_{n}(x) ~\propto~ 1+\log(x).\footnote{Recall from Section \ref{sec:Setup} that this aspect of the model is left unspecified by the researcher.}
\end{equation*}
Together with previous elements, this determines the true joint DGP $\Pi_{n}^{*}$ according to Eq.\ \eqref{eq:JointDistribution}.

Given any sample of observations $\{(a_{i},x_{i},x_{i}')\}_{i \leq n}$, the researcher estimates the parameters of interest $\theta = (\theta_{u,1},\theta_{u,2}, \theta_{f})$ using a two-step $K$-stage PI algorithm described in Sections \ref{sec:Inference}-\ref{sec:Applications}.  In the first step, the researcher estimates $P^*$ and $\theta_{f}^*$ using preliminary estimators $\hat{P}_n^{0} = \hat{P}_n$ and
\begin{equation*}
	\hat{\theta}_{f,n} ~=~\frac{\sum_{i=1}^n1[a_i=1,x_i'=x_i,x_i\ne |X|]}{\sum_{i=1}^n1[a_i=1,x_i\ne |X|]}.
\end{equation*}
In the second step, the researcher estimates $(\theta_{u,1},\theta_{u,2})$ using the $K$-stage policy function iteration algorithm using criterion function $Q_{n}$ equal to (a) pseudo-likelihood function ${Q}_{n}^{ML}$ in Section \ref{sec:ML} and (b) the weighted minimum distance function ${Q}_{n}^{MD}$ in Section \ref{sec:MD} with two limiting weight matrices: identity (i.e.\ $W^{*} = {\bf I}_{|\tilde{A} \times X| \times |\tilde{A} \times X|}$) and asymptotic variance minimizer (i.e.\ $W^{*} = W^{*}_{AV}$\footnote{This matrix is calculated using numerical derivates and Monte Carlo integration with a sample size that is significantly larger than those used in the actual Monte Carlo simulations.}). We show results for number of stages $K \in \{1,2,3,10\}$.\footnote{In accordance to our asymptotic theory, the simulation results with $K \in \{4,\dots,9\}$ are almost identical to those with $K \in \{3,10\}$. These were eliminated from the paper for reasons of brevity and are available from the authors upon request.}

We now describe simulation results for the estimator of $\theta_{u,2}$. We focus on $\theta_{u,2}$ because we consider the linear coefficient of the utility function to be more interesting than the constant coefficient.\footnote{The results for $\theta_{u,1}$ are qualitatively similar and are available from the authors upon request.} Table \ref{tab:miss0} describes results under correct specification. As expected, all estimators appear to converge to a distribution with zero mean and finite variance. Also as expected, the number of iterations $K$ does not seem to affect the bias or the variance of the estimators under consideration. Similar to \cite{aguirregabiria/mira:2002}, we detect small differences between the results with $K=1$ and those with $K>1$, especially for the smallest sample size. This effect tends to vanish as the sample size increases. These findings could be rationalized by the higher-order analysis in \cite{kasahara/shimotsu:2008}. The $K$-ML estimator and the $K$-MD estimator with $W^{*} = W^{*}_{AV}$ are similar and more efficient than the $K$-MD estimator with $W^{*} =  {\bf I}_{|\tilde{A} \times X| \times |\tilde{A} \times X|}$.

Table \ref{tab:miss1} provides results under asymptotically irrelevant local misspecification, i.e., $\delta=1$. According to our theoretical results, the asymptotic behavior of all estimators should be identical to the correctly specified model. This is confirmed by our simulations, as Tables \ref{tab:miss0} and \ref{tab:miss1} are virtually identical.

Table \ref{tab:miss2} provides results under local misspecification that vanishes at the knife-edge rate, i.e., $\delta=1/2$. According to our theoretical results, this should produce an asymptotic distribution that has non-zero bias and is not affected by the number of iterations $K$. By and large, these predictions are confirmed by our simulations. Given the presence of asymptotic bias, we now evaluate efficiency using the mean squared error. In this simulation design, the $K$-MD estimator with $W^{*} =  {\bf I}_{|\tilde{A} \times X| \times |\tilde{A} \times X|}$ is now also slightly more efficient than the $K$-ML estimator or the $K$-MD estimator with $W^{*} = W^{*}_{AV}$. This is also consistent with our theory: while $K$-MD estimator with $W^{*} =  {\bf I}_{|\tilde{A} \times X| \times |\tilde{A} \times X|}$ has more variance that the other two, it also appears to have less bias, resulting in less overall mean squared error.

Table \ref{tab:miss3b} provides results under asymptotically overwhelming local misspecification, i.e., $\delta=1/3$. According to our theoretical results, the presence of this local misspecification changes dramatically the asymptotic distribution of all estimators. In particular, these no longer converge at the regular $\sqrt{n}$-rate, but rather at the $n^{1/3}$-rate. In fact, at the $\sqrt{n}$-rate, the asymptotic bias is no longer bounded (See Table \ref{tab:miss3a} in the Supplement). Once we scale the estimators at the appropriate $n^{1/3}$-rate, they converge to an asymptotic distribution dominated by the bias. Furthermore, our theoretical results indicate that the number of iterations $K$ does not affect this asymptotic distribution. These predictions are clearly depicted in Table \ref{tab:miss3b}. In line with the results in Table \ref{tab:miss2}, the $K$-MD estimator with $W^{*} =  {\bf I}_{|\tilde{A} \times X| \times |\tilde{A} \times X|}$ is slightly more efficient than the $K$-ML estimator and the $K$-MD estimator with $W^{*} = W^{*}_{AV}$.

\begin{table}[H]
\begin{center}
\scalebox{0.85}
{
\begin{tabular}{cc|ccc|ccc|ccc}\hline
\hline
\multicolumn{1}{c}{\multirow{2}[4]{*}{$K$}} & \multicolumn{1}{c|}{\multirow{2}[4]{*}{Statistic}} & \multicolumn{3}{c|}{$K$-MD(${\bf I}_{|\tilde{A} \times X| \times |\tilde{A} \times X|}$)} & \multicolumn{3}{c|}{$K$-MD($W_{AV}^{*}$)} & \multicolumn{3}{c}{$K$-ML}\\
\multicolumn{1}{c}{} & \multicolumn{1}{c|}{} & 
\multicolumn{1}{c}{$n=200$} & \multicolumn{1}{c}{$n=500$} & \multicolumn{1}{c|}{$n=1,000$} & 
\multicolumn{1}{c}{$n=200$} & \multicolumn{1}{c}{$n=500$} & \multicolumn{1}{c|}{$n=1,000$} & 
\multicolumn{1}{c}{$n=200$} & \multicolumn{1}{c}{$n=500$} & \multicolumn{1}{c}{$n=1,000$} 
\\
\hline
  & $\sqrt{n}~$Bias &
   0.07 & 0.02 & 0.01 & 0.06 & 0.02 & 0.01 & 0.06 & 0.02 & 0.01 \\
$1$ & $\sqrt{n}~$SD &
   0.25 & 0.25 & 0.24 & 0.23 & 0.23 & 0.22 & 0.22 & 0.22 & 0.22 \\
  & $n~$MSE & 
   0.07 & 0.06 & 0.06 & 0.06 & 0.05 & 0.05 & 0.05 & 0.05 & 0.05 \\
   \hline
  & $\sqrt{n}~$Bias & 
   0.00 & 0.01 & 0.00 & 0.00 & 0.00 & 0.00 & 0.01 & 0.00 & 0.00 \\
$2$ & $\sqrt{n}~$SD &
   0.24 & 0.25 & 0.24 & 0.23 & 0.23 & 0.22 & 0.22 & 0.22 & 0.22 \\
  & $n~$MSE & 
   0.06 & 0.06 & 0.06 & 0.05 & 0.05 & 0.05 & 0.05 & 0.05 & 0.05 \\
  \hline
  & $\sqrt{n}~$Bias & 
   0.00 & 0.00 & 0.00 & 0.00 & 0.00 & 0.00 & 0.00 & 0.00 & 0.00 \\
$3$ & $\sqrt{n}~$SD & 
   0.24 & 0.25 & 0.24 & 0.23 & 0.23 & 0.22 & 0.22 & 0.22 & 0.22 \\
  & $n~$MSE & 
   0.06 & 0.06 & 0.06 & 0.05 & 0.05 & 0.05 & 0.05 & 0.05 & 0.05 \\
  \hline
  & $\sqrt{n}~$Bias & 
   0.00 & 0.00 & 0.00 & 0.00 & 0.00 & 0.00 & 0.00 & 0.00 & 0.00 \\
$10$ & $\sqrt{n}~$SD & 
   0.25 & 0.25 & 0.24 & 0.23 & 0.23 & 0.22 & 0.22 & 0.22 & 0.22 \\
  & $n~$MSE & 
   0.06 & 0.06 & 0.06 & 0.05 & 0.05 & 0.05 & 0.05 & 0.05 & 0.05 \\
\hline
\hline
\end{tabular}
	}\end{center}
	\caption{Simulation results under correct specification, i.e., $\tau_{n} = 0 $.}
	\label{tab:miss0}
\end{table}

\begin{table}[H]
\begin{center}
\scalebox{0.85}
{
\begin{tabular}{cc|ccc|ccc|ccc}\hline
\hline
\multicolumn{1}{c}{\multirow{2}[4]{*}{$K$}} & \multicolumn{1}{c|}{\multirow{2}[4]{*}{Statistic}} & \multicolumn{3}{c|}{$K$-MD(${\bf I}_{|\tilde{A} \times X| \times |\tilde{A} \times X|}$)} & \multicolumn{3}{c|}{$K$-MD($W_{AV}^{*}$)} & \multicolumn{3}{c}{$K$-ML}\\
\multicolumn{1}{c}{} & \multicolumn{1}{c|}{} & 
\multicolumn{1}{c}{$n=200$} & \multicolumn{1}{c}{$n=500$} & \multicolumn{1}{c|}{$n=1,000$} & 
\multicolumn{1}{c}{$n=200$} & \multicolumn{1}{c}{$n=500$} & \multicolumn{1}{c|}{$n=1,000$} & 
\multicolumn{1}{c}{$n=200$} & \multicolumn{1}{c}{$n=500$} & \multicolumn{1}{c}{$n=1,000$} 
\\
\hline
  & $\sqrt{n}~$Bias &0.11 & 0.04 & 0.03 & 0.10 & 0.04 & 0.03 & 0.09 & 0.04 & 0.03 \\
$1$ & $\sqrt{n}~$SD & 
0.25 & 0.25 & 0.24 & 0.24 & 0.23 & 0.22 & 0.23 & 0.23 & 0.22 \\
  & $n~$MSE & 
0.07 & 0.06 & 0.06 & 0.07 & 0.06 & 0.05 & 0.06 & 0.05 & 0.05 \\
   \hline
  & $\sqrt{n}~$Bias & 
0.04 & 0.03 & 0.02 & 0.04 & 0.03 & 0.02 & 0.04 & 0.03 & 0.02 \\
$2$ & $\sqrt{n}~$SD & 
0.25 & 0.25 & 0.24 & 0.23 & 0.23 & 0.22 & 0.22 & 0.22 & 0.22 \\
  & $n~$MSE & 
0.06 & 0.06 & 0.06 & 0.05 & 0.05 & 0.05 & 0.05 & 0.05 & 0.05 \\
  \hline
  & $\sqrt{n}~$Bias & 
0.04 & 0.03 & 0.02 & 0.04 & 0.03 & 0.02 & 0.04 & 0.03 & 0.02 \\
$3$ & $\sqrt{n}~$SD & 
0.25 & 0.25 & 0.24 & 0.23 & 0.23 & 0.22 & 0.22 & 0.22 & 0.22 \\
  & $n~$MSE & 
0.06 & 0.06 & 0.06 & 0.05 & 0.05 & 0.05 & 0.05 & 0.05 & 0.05 \\
  \hline
  & $\sqrt{n}~$Bias & 
0.04 & 0.03 & 0.02 & 0.04 & 0.03 & 0.02 & 0.04 & 0.03 & 0.02 \\
$10$ & $\sqrt{n}~$SD & 
0.25 & 0.25 & 0.24 & 0.23 & 0.23 & 0.22 & 0.22 & 0.22 & 0.22 \\
  & $n~$MSE & 
0.06 & 0.06 & 0.06 & 0.05 & 0.05 & 0.05 & 0.05 & 0.05 & 0.05 \\
\hline
\hline
\end{tabular}
	}\end{center}
	\caption{Simulation results under local misspecification with $\tau_{n} \propto n^{-1}$.}
	\label{tab:miss1}
\end{table}

\begin{table}[H]
\begin{center}
\scalebox{0.85}
{
\begin{tabular}{cc|ccc|ccc|ccc}\hline
\hline
\multicolumn{1}{c}{\multirow{2}[4]{*}{$K$}} & \multicolumn{1}{c|}{\multirow{2}[4]{*}{Statistic}} & \multicolumn{3}{c|}{$K$-MD(${\bf I}_{|\tilde{A} \times X| \times |\tilde{A} \times X|}$)} & \multicolumn{3}{c|}{$K$-MD($W_{AV}^{*}$)} & \multicolumn{3}{c}{$K$-ML}\\
\multicolumn{1}{c}{} & \multicolumn{1}{c|}{} & 
\multicolumn{1}{c}{$n=200$} & \multicolumn{1}{c}{$n=500$} & \multicolumn{1}{c|}{$n=1,000$} & 
\multicolumn{1}{c}{$n=200$} & \multicolumn{1}{c}{$n=500$} & \multicolumn{1}{c|}{$n=1,000$} & 
\multicolumn{1}{c}{$n=200$} & \multicolumn{1}{c}{$n=500$} & \multicolumn{1}{c}{$n=1,000$} 
\\
\hline
  & $\sqrt{n}~$Bias & 0.50 & 0.46 & 0.46 & 0.51 & 0.49 & 0.50 & 0.52 & 0.50 & 0.50 \\
$1$ & $\sqrt{n}~$SD &
0.28 & 0.28 & 0.27 & 0.26 & 0.26 & 0.24 & 0.25 & 0.25 & 0.24 \\
  & $n~$MSE & 
0.33 & 0.29 & 0.28 & 0.33 & 0.31 & 0.30 & 0.33 & 0.31 & 0.31 \\
  \hline
  & $\sqrt{n}~$Bias & 
0.42 & 0.44 & 0.45 & 0.46 & 0.48 & 0.49 & 0.47 & 0.49 & 0.49 \\
$2$ & $\sqrt{n}~$SD &
0.29 & 0.28 & 0.26 & 0.26 & 0.25 & 0.24 & 0.24 & 0.24 & 0.24 \\
  & $n~$MSE & 
0.26 & 0.27 & 0.28 & 0.27 & 0.29 & 0.30 & 0.28 & 0.30 & 0.30 \\
  \hline
  & $\sqrt{n}~$Bias & 
0.42 & 0.44 & 0.45 & 0.46 & 0.48 & 0.49 & 0.47 & 0.49 & 0.49 \\
$3$ & $\sqrt{n}~$SD & 
0.28 & 0.28 & 0.26 & 0.26 & 0.25 & 0.24 & 0.24 & 0.24 & 0.24 \\
  & $n~$MSE & 
0.26 & 0.27 & 0.28 & 0.27 & 0.29 & 0.30 & 0.28 & 0.30 & 0.30 \\
  \hline
  & $\sqrt{n}~$Bias & 
0.42 & 0.44 & 0.45 & 0.46 & 0.48 & 0.49 & 0.47 & 0.49 & 0.49 \\
$10$ & $\sqrt{n}~$SD & 
0.29 & 0.28 & 0.26 & 0.26 & 0.25 & 0.24 & 0.24 & 0.24 & 0.24 \\
  & $n~$MSE &
0.26 & 0.27 & 0.28 & 0.27 & 0.29 & 0.30 & 0.28 & 0.30 & 0.30 \\
\hline
\hline
\end{tabular}
	}\end{center}
	\caption{Simulation results under local misspecification with $\tau_{n} \propto n^{-1/2}$.}
	\label{tab:miss2}
\end{table}

\begin{table}[H]
\begin{center}
\scalebox{0.85}
{
\begin{tabular}{cc|ccc|ccc|ccc}\hline
\hline
\multicolumn{1}{c}{\multirow{2}[4]{*}{$K$}} & \multicolumn{1}{c|}{\multirow{2}[4]{*}{Statistic}} & \multicolumn{3}{c|}{$K$-MD(${\bf I}_{|\tilde{A} \times X| \times |\tilde{A} \times X|}$)} & \multicolumn{3}{c|}{$K$-MD($W_{AV}^{*}$)} & \multicolumn{3}{c}{$K$-ML}\\
\multicolumn{1}{c}{} & \multicolumn{1}{c|}{} & 
\multicolumn{1}{c}{$n=200$} & \multicolumn{1}{c}{$n=500$} & \multicolumn{1}{c|}{$n=1,000$} & 
\multicolumn{1}{c}{$n=200$} & \multicolumn{1}{c}{$n=500$} & \multicolumn{1}{c|}{$n=1,000$} & 
\multicolumn{1}{c}{$n=200$} & \multicolumn{1}{c}{$n=500$} & \multicolumn{1}{c}{$n=1,000$} 
\\
\hline
  & $n^{1/3}~$Bias & 0.41 & 0.40 & 0.41 & 0.43 & 0.44 & 0.45 & 0.45 & 0.46 & 0.46 \\
$1$ & $n^{1/3}~$SD &
0.14 & 0.12 & 0.10 & 0.13 & 0.10 & 0.09 & 0.12 & 0.10 & 0.09 \\
  & $n^{2/3}~$MSE & 
0.18 & 0.18 & 0.18 & 0.20 & 0.20 & 0.21 & 0.22 & 0.22 & 0.22 \\
\hline
  & $n^{1/3}~$Bias & 
0.37 & 0.40 & 0.41 & 0.40 & 0.43 & 0.45 & 0.43 & 0.45 & 0.46 \\
$2$ & $n^{1/3}~$SD &
0.14 & 0.12 & 0.10 & 0.12 & 0.10 & 0.09 & 0.11 & 0.10 & 0.09 \\
  & $n^{2/3}~$MSE & 
0.16 & 0.17 & 0.18 & 0.18 & 0.20 & 0.21 & 0.20 & 0.21 & 0.22 \\
\hline
  & $n^{1/3}~$Bias & 
0.37 & 0.40 & 0.41 & 0.40 & 0.43 & 0.45 & 0.43 & 0.45 & 0.46 \\
$3$ & $n^{1/3}~$SD & 
0.14 & 0.12 & 0.10 & 0.12 & 0.10 & 0.09 & 0.11 & 0.10 & 0.09 \\
  & $n^{2/3}~$MSE &
0.16 & 0.17 & 0.18 & 0.18 & 0.20 & 0.21 & 0.20 & 0.21 & 0.22 \\
\hline
  & $n^{1/3}~$Bias & 
0.37 & 0.40 & 0.41 & 0.40 & 0.43 & 0.45 & 0.43 & 0.45 & 0.46 \\
$10$ & $n^{1/3}~$SD & 
0.14 & 0.12 & 0.10 & 0.12 & 0.10 & 0.09 & 0.11 & 0.10 & 0.09 \\
  & $n^{2/3}~$MSE & 
0.16 & 0.17 & 0.18 & 0.18 & 0.20 & 0.21 & 0.20 & 0.21 & 0.22 \\
\hline
\hline
\end{tabular}
	}\end{center}
	\caption{Simulation results under local misspecification with $\tau_{n} \propto n^{-1/3}$ and using the correct scaling.}
	\label{tab:miss3b}
\end{table}

\section{Conclusion}\label{sec:Conclusion} 

Single-agent dynamic discrete choice models are typically estimated using heavily parametrized econometric frameworks, making them susceptible to model misspecification. This paper investigates how misspecification can affect inference results in these models. This paper considers a {\it local misspecification} framework, which is an asymptotic device in which the mistake in the specification vanishes as the sample size diverges. In this paper, we impose no restrictions on the rate at which these specification errors disappear. Relative to global misspecification, the local misspecification analysis has two important advantages. First, it yields tractable and general results. Second, it allows us to focus on parameters with structural interpretation, instead of ``pseudo-true'' parameters.

We consider a general class of two-step estimators based on the $K$-stage sequential policy function iteration algorithm, where $K$ denotes the number of iterations employed in the estimation. By appropriate choice of the criterion function, this class includes \cite{hotz/miller:1993}'s conditional choice probability estimator, \cite{aguirregabiria/mira:2002}'s pseudo-likelihood estimator, and \cite{pesendorfer/schmidt-dengler:2008}'s asymptotic least squares estimator.

We show that local misspecification can affect the asymptotic distribution and even the rate of convergence of these estimators. In principle, one might expect that the effect of the local misspecification could change with the number of iterations $K$. The main finding in the paper is that this is not the case, i.e., the effect of local misspecification is invariant to $K$. In particular, (a) $K$-ML estimators are asymptotically equivalent and (b) given the choice of the weight matrix, $K$-MD estimators are asymptotically equivalent. In practice, this means that researchers cannot eliminate or even alleviate problems of model misspecification by changing $K$. Additional iterations are computationally costly and produce {\it no change} in asymptotic efficiency.

Under correct specification, the comparison between $K$-MD and $K$-ML estimators in terms of asymptotic mean squared error yields a clear-cut recommendation. Under local misspecification, this is no longer the case. In particular, local misspecification can introduce an unknown asymptotic bias, which complicates this comparison. In the presence of asymptotic bias, the optimality of the estimator should be evaluated using the asymptotic mean squared error. We show that an optimally-weighted $K$-MD estimator depends on the unknown asymptotic bias and is thus generally unfeasible. In turn, the feasible $K$-MD estimator with a weight matrix that minimizes asymptotic variance could have an asymptotic mean squared error that is higher or lower than that of the $K$-ML estimator or the $K$-MD estimator with identity weight matrix.

% \newpage
\appendix

{
 \small
\section{Appendix}

%%%%%%%% DIVIDER %%%%%%%%%%%%

\subsection{Additional notation}

Throughout this appendix, ``s.t.'' abbreviates ``such that'', and ``RHS'' and `` LHS'' abbreviate ``right-hand side'' and ``left-hand side'', respectively. Furthermore, ``LLN'' refers to the strong law of large numbers, ``CLT'' refers to the central limit theorem, and ``CMT'' refers to the continuous mapping theorem.

Given the true DGP $\Pi_{n}^{*}$, Eq.\ \eqref{eq:DefnPopulations} defined transition probabilities $f^{*}_{n}$, CCPs $P^{*}_{n}$, and marginal distribution of states $m^{*}_{n}$. The unconditional probability of $(a,x) \in A \times X$ is analogously defined by:
	\begin{eqnarray*}
		{J}_{n}^{*}(a,x) ~\equiv~ 
			{\sum_{\tilde{x}^{\prime }\in X} \Pi_{n}^{\ast }( a,x,\tilde{x} ^{\prime })},
	\end{eqnarray*}
and $J_n^{*} \equiv \{{J}_{n}^{*}(a,x)\}_{(a,x)\in A \times X}$. 
The limiting DGP $f^{*}$, transition probabilities $f^{*}$, CCPs $P^{*}$ were defined in Assumption \ref{ass:LocalMiss}. The other limiting objects are analogously defined by $J^{*} \equiv \lim_{n\to \infty} J^{*}_{n}$ and $m^{*} \equiv \lim_{n\to \infty} m^{*}_{n}$.
	
The sample analogue DGP $\hat{\Pi}_{n}$, transition probabilities $\hat{f}_{n}$, CCPs $\hat{P}_{n}$ were defined in Eq.\ \eqref{eq:DefnSamples}. The sample analogue marginal distribution of states $\hat{m}_{n}$ and unconditional probabilities $\hat{J}_{n}$ are analogously defined. For any $(a,x) \in A \times X$, 
\begin{align*}
	\hat{m}_{n}(x) ~&\equiv~ 
		{\sum_{(a,\tilde{x}^{\prime })\in A\times X} \hat{\Pi}_{n}( a,x,\tilde{x} ^{\prime })}\\
	\hat{J}_{n}(a,x) ~&\equiv~ 
		{\sum_{\tilde{x}^{\prime }\in X} \hat{\Pi}_{n}( a,x,\tilde{x} ^{\prime })},
\end{align*}
$\hat{m}_{n} \equiv \{\hat{m}_{n}(x)\}_{x \in X}$, and $\hat{J}_{n} \equiv \{\hat{J}_{n}(a,x)\}_{(a,x)\in A \times X}$.

%%%%%%%% DIVIDER %%%%%%%%%%%%

\subsection{Proofs of theorems}

%%%%%%%% DIVIDER %%%%%%%%%%%%

\begin{proof}[Proof of Theorem \ref{thm:LocalMisspecification}]
Since $\Theta =\Theta _{\alpha }\times \Theta _{f}$ is compact and $||(P_{(\alpha ,\theta _{f})}-P_{n}^{\ast })',(f_{\theta_{f}}-f_{n}^{\ast })'||$ is a continuous function of $(\alpha ,\theta _{f})$, the arguments in \citet[pages 193-195]{royden:1988} implies that $\exists (\alpha ^{\ast },\theta _{f}^{\ast })\in \Theta $ that minimizes $||(P_{(\alpha ,\theta _{f})}-P_{n}^{\ast })',(f_{\theta_{f}}-f_{n}^{\ast })'||$. By Assumption \ref{ass:LocalMiss}(b), this minimum value is zero, i.e., $\exists (\alpha ^{\ast },\theta _{f}^{\ast })\in \Theta $ s.t. $||(P_{(\alpha ^{\ast },\theta _{f}^{\ast })}-P^{\ast })',(f_{\theta _{f}^{\ast }}-f^{\ast })'||=0$ or, equivalently, $P_{(\alpha ^{\ast },\theta _{f}^{\ast })}=P^{\ast }$ and $f_{\theta _{f}^{\ast }}=f^{\ast }$. 

Now suppose that this also occurs for $( \tilde{\theta}_{f},\tilde{\alpha})\in \Theta $. We now show that $(\theta _{f}^{\ast },\alpha ^{\ast })=(\tilde{ \theta}_{f},\tilde{\alpha})$. By triangle inequality $||f_{\theta _{f}^{\ast }}-f_{\tilde{\theta}_{f}}||\leq ||f_{\theta _{f}^{\ast }}-f^{\ast }||+||f_{ \tilde{\theta}_{f}}-f^{\ast }\Vert $ and since $\theta _{f}^{\ast }$ and $ \tilde{\theta}_{f}$ both satisfy $||f_{\theta _{f}}-f^{\ast }||=0$, we conclude that $||f_{\theta _{f}^{\ast }}-f_{\tilde{\theta}_{f}}||=0$ and so $ f_{\theta _{f}^{\ast }}=f_{\tilde{\theta}_{f}}$. By Assumption \ref{ass:Identification}, this implies that $\theta _{f}^{\ast }=\tilde{\theta} _{f}$. By repeating the previous argument with $P_{(\alpha ,\theta _{f}^{\ast })}$ instead of $f_{\theta _{f}}$, we conclude that $\alpha ^{\ast }= \tilde{\alpha}$.
\end{proof}

%%%%%%%% DIVIDER %%%%%%%%%%%%

\begin{proof}[Proof of Theorem \ref{thm:ANproof}]
	Without loss of generality, we can consider the neighborhood $\mathcal{N}$ to be ``rectangular'', in the sense that $\mathcal{N} = \mathcal{N}_{\alpha} \times \mathcal{N}_{\theta_{f}} \times \mathcal{N}_{P}$, where $\mathcal{N}_{\lambda}$ denotes the neighborhood of $\lambda^{*}$ for any $\lambda \in \{\alpha,\theta_f,P\}$. This can always be achieved by replacing $\mathcal{N}$ with $\tilde{\mathcal{N}}\subseteq \mathcal{N}$ that has the desired structure. This proof will make repeated reference to $\mathcal{N}_{\alpha }$.

\underline{Part 1.} Fix $K\geq 1$ arbitrarily. We prove the result by assuming that:
\begin{equation}
n^{\min \{\delta ,1/2\}}(\hat{P}_{n}^{K-1}-P^{\ast })=O_{p_{n}}(1). \label{eq:Distribution}
\end{equation}

By definition, $\hat{\alpha}_{n}^{K} = \arg \max_{\alpha \in \Theta _{\alpha }}Q_{n}(\alpha )$ with $Q_{n}(\alpha ) \equiv {Q}_{n}(\alpha ,\hat{\theta} _{f,n},\hat{P}_{n}^{K-1})$ and $\hat{P}_{n}^{K-1}\equiv\Psi_{(\hat{\alpha} _{n}^{K-1},\hat{\theta}_{f,n})}(\hat{P}_{n}^{K-2})$ for $K>1$ and $\hat{P}_{n}^{K-1}\equiv\hat{P}_{n}^{0}$ for $K=1$. The result then follows from Theorem \ref{thm:ANforEE}. To apply this result, we first check its conditions.

Condition (a). Under Assumptions \ref{ass:EEpart2}(a)-(c) and \ref{ass:EEpart3}, Theorem \ref{thm:consistencyEE} implies that $\hat{ \alpha}_{n}^{K}=\alpha ^{\ast }+o_{p_{n}}(1)$.

Condition (b). This is imposed in Assumption \ref{ass:EEpart1}.

Assumption \ref{ass:EEpart3}(a) and Eq.\ \eqref{eq:Distribution} imply that $ (\alpha ^{\ast },\hat{\theta}_{f,n},\hat{P}^{K-1}_{n})\in \mathcal{N}$ w.p.a.1. In turn, this and $\hat{ \alpha}_{n}^{K}=\alpha ^{\ast }+o_{p_{n}}(1)$ imply that $(\hat{ \alpha}_{n}^{K},\hat{\theta}_{f,n},\hat{P}^{K-1}_{n})\in \mathcal{N}$ w.p.a.1. These results will be used repeatedly throughout the rest of this proof.

Condition (c). This follows from Assumption \ref{ass:EEpart2}(c) and $(\hat{ \alpha}_{n}^{K},\hat{\theta}_{f,n},\hat{P}^{K-1}_{n})\in \mathcal{N}$ w.p.a.1.

Condition (d). Assumptions \ref{ass:EEpart2}(d)-(f), \ref{ass:EEpart3}(b), and $(\hat{ \alpha}_{n}^{K},\hat{\theta}_{f,n},\hat{P}^{K-1}_{n})\in \mathcal{N}$ w.p.a.1 imply that the following derivation holds w.p.a.1.
\begin{align*}
 n^{\min \{\delta ,1/2\}}\frac{\partial Q_{n}(\alpha ^{\ast })}{\partial \alpha } &=n^{\min \{\delta ,1/2\}}\frac{\partial {Q}_{n}(\alpha ^{\ast }, \hat{\theta}_{f,n},\hat{P}_{n}^{K-1})}{\partial \alpha } \\
&=n^{\min \{\delta ,1/2\}}\frac{\partial {Q}_{n}(\alpha ^{\ast },\theta _{f}^{\ast },P^{\ast })}{\partial \alpha }+\frac{\partial^2 {Q}_{n}(\alpha ^{\ast },\tilde{\theta}_{f,n},\tilde{P}_{n})}{\partial \alpha \partial \theta _{f}^{\prime }}n^{\min \{\delta ,1/2\}}(\hat{\theta}_{f,n}-\theta _{f}^{\ast })\\
&+\frac{\partial^2 {Q}_{n}(\alpha ^{\ast },\tilde{\theta}_{f,n}, \tilde{P}_{n})}{\partial \alpha \partial P^{\prime }}n^{\min \{\delta ,1/2\}}(\hat{P}_{n}^{K-1}-P^{\ast }) \\
&=n^{\min \{\delta ,1/2\}}\left[\frac{\partial {Q}_{n}(\alpha ^{\ast },\theta _{f}^{\ast },P^{\ast })}{\partial \alpha }+\frac{\partial^2 {Q}_{\infty }(\alpha ^{\ast },\theta _{f}^{\ast },P^{\ast })}{\partial \alpha \partial \theta _{f}'}(\hat{\theta}_{f,n}-\theta _{f}^{\ast })\right]+o_{p_{n}}(1),
\end{align*}
where $(\tilde{\theta}_{f,n},\tilde{P}_{n})$ is some sequence between $(\hat{ \theta}_{f,n},\hat{P}_{n}^{K-1})$ and $(\theta _{f}^{\ast },P^{\ast })$. From this and Assumption \ref{ass:EEpart3}(a),
\begin{equation*}
n^{\min \{\delta ,1/2\}}\frac{\partial Q_{n}(\alpha ^{\ast })}{\partial \alpha }~\overset{d}{\to }~\zeta _{1}+\frac{\partial^2 {Q}_{\infty }(\alpha ^{\ast },\theta _{f}^{\ast },P^{\ast })}{\partial \alpha \partial \theta _{f}^{\prime }}\zeta _{2}.
\end{equation*}
If we denote the RHS random variable by $Z$, condition (d) follows.

Condition (e)-(f). Consider any arbitrary $\alpha \in \mathcal{N}_{\alpha }$ and so $(\alpha,\hat{\theta}_{f,n},\hat{P}^{K-1}_{n}) \in \mathcal{N}$ w.p.a.1. This and Assumptions \ref{ass:EEpart2}(c)-(e) imply that the following derivation holds w.p.a.1.
\begin{equation*}
\frac{\partial^2 Q_{n}(\alpha )}{\partial \alpha \partial \alpha ^{\prime }}= \frac{\partial^2 {Q}_{n}(\alpha ,\hat{\theta}_{f,n},\hat{P}_{n}^{K-1})}{ \partial \alpha \partial \alpha ^{\prime }}=\frac{\partial^2 {Q}_{\infty }(\alpha ,\hat{\theta}_{f_{n}},\hat{P}_{n}^{K-1})}{\partial \alpha \partial \alpha ^{\prime }}+o_{p_{n}}(1)=\frac{\partial^2 {Q}_{\infty }(\alpha ,\theta _{f}^{\ast },P^{\ast })}{\partial \alpha \partial \alpha ^{\prime }} +o_{p_{n}}(1),
\end{equation*}
where convergence is uniform in $\alpha \in \mathcal{N}_{\alpha }$, i.e., condition (e) follows. In addition, if we denote the first term on the RHS by $H(\alpha)$, condition (f) follows.

Under these conditions, Theorem \ref{thm:ANforEE} then implies that:
\begin{equation}
n^{\min \{\delta ,1/2\}}(\hat{\alpha}_{n}^{K}-\alpha ^{\ast })~=~A_{1}n^{\min \{\delta ,1/2\}}\frac{\partial {Q}_{n}(\alpha ^{\ast },\theta _{f}^{\ast },P^{\ast })}{\partial \alpha}+A_{2}n^{\min \{\delta ,1/2\}}(\hat{ \theta}_{f,n}-\theta _{f}^{\ast })+o_{p_{n}}(1), \label{eq:FormulaAlpha}
\end{equation}
with
\begin{align*}
A_{1}~&\equiv~ -\left(\frac{\partial^2 {Q}_{\infty }(\alpha ,\theta _{f}^{\ast },P^{\ast })}{\partial \alpha \partial \alpha ^{\prime }}\right)^{-1}\\
A_{2}~&\equiv~ -\left(\frac{\partial^2 {Q}_{\infty }(\alpha ,\theta _{f}^{\ast },P^{\ast })}{\partial \alpha \partial \alpha ^{\prime }}\right)^{-1}\frac{ \partial^2 {Q}_{\infty }(\alpha ^{\ast },\theta _{f}^{\ast },P^{\ast })}{ \partial \alpha \partial \theta _{f}^{\prime }}.
\end{align*}

\underline{Part 2.} The objective of this part is to show Eq.\ \eqref{eq:Distribution} holds for all $K\geq 1$. We show this by induction.

Initial step. For $K=1$, the result holds by Assumption \ref{ass:EEpart3}(b). Also, part 1 implies that Eq.\ \eqref{eq:FormulaAlpha} holds for $K=1$.

Inductive step. The inductive assumption is that Eqs.\ \eqref{eq:Distribution}-\eqref{eq:FormulaAlpha} hold for some $K\geq 1$. Our goal is then to show that Eqs.\ \eqref{eq:Distribution}-\eqref{eq:FormulaAlpha} hold with $K$ replaced by $K+1$. By inductive assumption, $(\hat{\alpha}_{n}^{K},\hat{\theta}_{f,n},\hat{P}_{n}^{K-1})=(\alpha ^{\ast },\theta _{f}^{\ast },P^{\ast })+o_{p_{n}}(1)$ and so $(\hat{\alpha}_{n}^{K}, \hat{\theta}_{f,n},\hat{P}_{n}^{K-1})\in \mathcal{N}$ w.p.a.1. Then, the following derivation holds:
\begin{align*}
n^{\min \{\delta ,1/2\}}(\hat{P}_{n}^{K}-P^{\ast })  &=n^{\min \{\delta ,1/2\}}(\Psi_{(\hat{\alpha}_{n}^{K},\hat{\theta} _{f,n})}(\hat{P}_{n}^{K-1})-\Psi_{(\alpha ^{\ast },\theta _{f}^{\ast })}(P^{\ast })) \\
&=\frac{\partial \Psi_{(\hat{\alpha}_{n}^{K},\hat{\theta}_{f,n})}(\hat{P} _{n}^{K-1})}{\partial \alpha ^{\prime }}n^{\min \{\delta ,1/2\}}(\hat{\alpha} _{n}^{K}-\alpha  ^{\ast }) +\frac{\partial \Psi_{(\tilde{\alpha}_{n}^{K},\tilde{\theta} _{f,n})}(\tilde{P}_{n}^{K-1})}{\partial \theta _{f}^{\prime }}n^{\min \{\delta ,1/2\}}(\hat{\theta}_{f,n}-\theta ^{\ast })\\
&+\frac{\partial \Psi_{(\tilde{ \alpha}_{n}^{K},\tilde{\theta}_{f,n})}(\tilde{P}_{n}^{K-1})}{\partial P^{\prime }}n^{\min \{\delta ,1/2\}}(\hat{P}_{n}^{K-1}-P^{\ast })\\
&=\frac{\partial \Psi_{(\alpha ^{\ast },\theta _{f}^{\ast })}(P^{\ast })}{\partial \alpha ^{\prime }} n^{\min \{\delta ,1/2\}}(\hat{\alpha}_{n}^{K}-\alpha ) +\frac{ \partial \Psi_{(\alpha ^{\ast },\theta _{f}^{\ast })}(P^{\ast })}{\partial \theta _{f}^{\prime }} n^{\min \{\delta ,1/2\}}(\hat{\theta}_{f,n}-\theta ^{\ast })+o_{p_{n}}(1), 
\end{align*}
where $(\tilde{\alpha}_{n},\tilde{\theta}_{f,n},\tilde{P}_{n})$ is between $(\hat{\alpha} _{n}^{K},\hat{ \theta}_{f,n},\hat{P}_{n}^{K-1})$ and $(\alpha^{\ast },\theta _{f}^{\ast },P^{\ast })$, and the first equality uses that $P^{\ast }=\Psi_{(\alpha ^{\ast },\theta _{f}^{\ast })}(P^{\ast })$, and the final equality holds by Lemma \ref{lem:PolicyProperties}(c) and $n^{\min \{\delta ,1/2\}}((\hat{\alpha}_{n}^{K},\hat{\theta}_{f,n},\hat{P}_{n}^{K-1}) - (\alpha ^{\ast },\theta _{f}^{\ast },P^{\ast })) = O_{p_{n}}(1)$. From this, we conclude that $n^{\min \{\delta ,1/2\}}(\hat{P} _{n}^{K}-P^{\ast })=O_{p_{n}}(1)$, i.e., Eq.\ \eqref{eq:Distribution} holds with $K$ replaced by $K+1$. In turn, this and part 1 then imply that Eq.\ \eqref{eq:FormulaAlpha} holds with $K$ replaced by $K+1$. This concludes the inductive step and the proof.
\end{proof}

%%%%%%%% DIVIDER %%%%%%%%%%%%

\begin{proof}[Proof of Theorem \ref{thm:ML_main}]
This result is a corollary of Theorem \ref{thm:ANproof} and Lemma \ref{lem:AuxResult2}. To apply Theorem \ref{thm:ANproof}, we first need to verify Assumptions \ref{ass:EEpart2}-\ref{ass:EEpart3}. We anticipate that $Q_{\infty }^{ML}(\theta ,P)=\sum_{(a,x)\in A \times X}J^{\ast }(a,x)\ln \Psi_{\theta }(P)(a|x)$. 

\underline{Part 1:} Verify Assumption \ref{ass:EEpart2}.

Condition (a). First, notice that $\hat{J}_{n}-J^{\ast }=o_{p_{n}}(1)$ and $ \Psi_{\theta }(P)(a|x)>0$ for all $(\theta ,P)\in \Theta \times \Theta _{P}$ and $(a,x)\in A \times X$ implies that $Q_{n}^{ML}(\theta ,P)-Q_{\infty }^{ML}(\theta ,P)=o_{p_{n}}(1)$. Furthermore, notice that:
\begin{align*}
&\sup_{(\theta ,P)\in \Theta \times \Theta _{P}}\vert Q_{n}^{ML}(\theta ,P)-Q_{\infty }^{ML}(\theta ,P)\vert 
=\sup_{(\theta ,P)\in \Theta \times \Theta _{P}}\left\vert \sum_{(a,x)\in A \times X}( \hat{J} _{n}(a,x)-J^{\ast }(a,x)) \ln \Psi_{\theta }(P)(a|x)\right\vert \\
&\leq  \sum_{(a,x)\in A \times X} \vert \hat{J}_{n}(a,x)-J^{\ast }(a,x) \vert \times \left\vert \ln \left( \min_{(a,x)\in A \times X}\inf_{( \theta ,P) \in \Theta \times \Theta _{P}}\Psi_{\theta }(P)(a|x)\right) \right\vert
\end{align*}
Since $\Psi_{\theta }(P)(a|x)>0$ for all $(\theta ,P)\in \Theta \times \Theta _{P}$ and $(a,x)\in A \times X$, $\Psi_{\theta }(P)(a|x):\Theta \times \Theta _{P}\to \mathbb{R} $ is continuous in $(\theta ,P)$ for all $(a,x)\in A \times X$, and $ \Theta \times \Theta _{P}$ is compact, $\min_{(a,x)\in A \times X}\inf_{( \theta ,P) \in \Theta \times \Theta _{P}}\Psi_{\theta }(P)(a|x)>0 $. From this and $\hat{J}_{n}-J^{\ast }=o_{p_{n}}(1)$, we conclude that $ \sup_{(\theta ,P)\in \Theta \times \Theta _{P}}\vert Q_{n}^{ML}(\theta ,P)-Q_{\infty }^{ML}(\theta ,P)\vert =o_{p_{n}}(1)$. Second, previous arguments imply that $Q_{\infty }^{ML}(\theta ,P):\Theta \times \Theta _{P}\to \mathbb{R} $ is continuous in $(\theta ,P)$. Since $\Theta \times \Theta _{P}$ is compact, it then follows that $ Q_{\infty }^{ML}(\theta ,P):\Theta \times \Theta _{P}\to \mathbb{R} $ is uniformly continuous in $(\theta ,P)$. Third, $(\hat{\theta}_{f,n},\tilde{P} _{n})-(\theta _{f}^{\ast },P^{\ast })=o_{p_{n}}(1)$, where $\tilde{P} _{n}$ is the arbitrary sequence in condition (a). By combining these with \citet[Lemma 24.1]{gourieroux/monfort:1995b}, the result follows.

Condition (b). This follows from Assumption \ref{ass:Identification} and the information inequality (e.g.\ \citet[Theorem 2.3]{white:1996}).

Condition (c). Since $\Psi_{\theta }(P)(a|x)>0$ for all $(\theta ,P)\in \Theta \times \Theta _{P}$ and $(a,x)\in A \times X$, and $\Psi_{\theta }(P)(a|x):\Theta \times \Theta _{P}\to \mathbb{R} $ is twice continuously differentiable in $(\theta ,P)$ for all $(a,x)\in A \times X$, $\ln \Psi_{\theta }(P)(a|x):\Theta \times \Theta _{P}\to \mathbb{R} $ is twice continuously differentiable in $(\theta ,P)$ for all $(a,x)\in A \times X$. From here, the result follows.

Condition (d). By direct computation,
\begin{align*}
&\sup_{(\theta ,P)\in \mathcal{N}}\left\vert \frac{\partial^2 Q_{n}^{ML}(\alpha ,\theta _{f},P)}{\partial \alpha \partial \lambda' }-\frac{\partial^2 Q_{\infty }^{ML}(\alpha ,\theta _{f},P)}{\partial \alpha \partial \lambda'}\right\vert=\sup_{(\theta ,P)\in \mathcal{N}}\left\vert \sum_{(a,x)\in A \times X}( J_{n}^{\ast }(a,x)-J^{\ast }(a,x)) M_{\theta ,P}(a,x)\right\vert \\
&\leq  \sum_{(a,x)\in A \times X}\vert J_{n}^{\ast }(a,x)-J^{\ast }(a,x)\vert \times  \max_{( a,x) \in A \times X}\sup_{(\theta ,P)\in \Theta \times \Theta _{P}}\vert M_{\theta ,P}(a,x)\vert,
\end{align*}
with:
\[
M_{\theta ,P}(a,x) ~\equiv~ \frac{-1}{( \Psi_{\theta }(P)(a|x)) ^{2}} \frac{\partial \Psi_{\theta }(P)(a|x)}{\partial \alpha}\frac{ \partial \Psi_{\theta }(P)(a|x)}{\partial \lambda' }+\frac{1}{\Psi_{\theta }(P)(a|x)}\frac{ \partial^2 \Psi_{\theta }(P)(a|x)}{\partial \alpha \partial \lambda'}.
\]
Since $\Psi_{\theta }(P)(a|x)>0$ for all $(\theta ,P)\in \Theta \times \Theta _{P}$ and $(a,x)\in A \times X$, $\Psi_{\theta }(P)(a|x):\Theta \times \Theta _{P}\to \mathbb{R} $ is continuous in $(\theta ,P)$ for all $(a,x)\in A \times X$, and $ \Theta \times \Theta _{P}$ is compact, $\inf_{( \theta ,P) \in \Theta \times \Theta _{P}}\Psi_{\theta }(P)(a|x)>0$ for all $(a,x)\in A \times X$. From this and that $\Psi_{\theta }(P)$ is twice continuously differentiable in $(\theta ,P)$, $M_{\theta ,P}(a,x)$ is continuous in $(\theta ,P)$ for all $(a,x)\in A \times X$. Since $\Theta \times \Theta _{P} $ is compact, $\max_{( a,x) \in A \times X}\sup_{(\theta ,P)\in \Theta \times \Theta _{P}}\vert M_{\theta ,P}(a,x)\vert <\infty $. From this and $\hat{J} _{n}-J^{\ast }=o_{p_{n}}(1)$, the result follows.

Condition (e). By direct computation, for any $\lambda \in (\alpha ,\theta _{f},P)$,
\begin{equation}
\tfrac{\partial^2 Q_{\infty }^{ML}(\alpha ,\theta _{f},P)}{\partial \alpha \partial \lambda' }=
\sum_{(a,x)\in A \times X}J^{\ast }(a,x)
\left[\frac{-1}{ (\Psi_{\theta }(P)(a,x))^{2}}\frac{\partial \Psi_{\theta }(P)(a|x)}{ \partial \alpha }\frac{\partial \Psi_{\theta }(P)(a|x)}{\partial \lambda'}+\frac{1}{\Psi_{\theta }(P)(a|x)}\frac{\partial^2 \Psi_{\theta }(P)(a|x)}{\partial \alpha \partial  \lambda'}\right].
\label{eq:SecondDeriv}
\end{equation}
This function is continuous and if we evaluate it at $(\alpha ,\theta _{f},P)=(\alpha ^{\ast },\theta _{f}^{\ast },P^{\ast })$, we obtain:
\begin{align}
&\frac{\partial^2 Q_{\infty }^{ML}(\alpha ^{\ast },\theta _{f}^{\ast },P^{\ast })}{\partial \alpha \partial \lambda' } =\sum_{(a,x)\in A \times X}J^{\ast }(a,x)\left[\frac{-1}{(P_{\theta ^{\ast }}(a|x))^2} \frac{\partial P_{\theta ^{\ast }}(a|x)}{\partial \alpha }\frac{\partial P_{\theta ^{\ast }}(a|x)}{\partial \lambda' }+\frac{1}{P_{\theta ^{\ast }}(a|x)}\frac{\partial^2 \Psi_{\theta^{\ast} }(P^{\ast})(a|x)}{\partial \alpha \partial \lambda' }\right] \nonumber \\
&= - \sum_{(a,x)\in A \times X}J^{\ast }(a,x)  \frac{\partial \ln P_{\theta ^{\ast }}(a|x)}{\partial \alpha }\frac{\partial \ln P_{\theta ^{\ast }}(a|x)}{\partial \lambda'} +  \sum_{x\in X}m^{\ast }(x) \sum_{a\in A} \frac{\partial^2 \Psi_{\theta^{\ast} }(P^{\ast})(a|x)}{\partial \alpha \partial \lambda'} \nonumber \\
& = -\sum_{(a,x)\in A \times X}J^{\ast }(a,x)\frac{\partial \ln P_{\theta ^{\ast }}(a|x)}{\partial \alpha}\frac{\partial \ln P_{\theta ^{\ast }}(a|x)}{ \partial \lambda} = - \frac{\partial P_{\theta^*}'}{\partial \alpha} \Phi \frac{\partial P_{\theta^*}}{\partial \lambda'}
\label{eq:SecondDerivMLE}
\end{align}
where the first equality uses $ \partial \Psi_{\theta ^{\ast }}(P^{\ast })/\partial \alpha =\partial P_{\theta ^{\ast }}/\partial \alpha $ and $P_{\theta ^{\ast }}=P^{\ast }$, the second equality uses $J^{\ast }(a,x)=P^{\ast }(a|x)m^{\ast }(x)$, the third equality interchanges summation and differentiation and uses that $\sum_{a\in A}\Psi_{\theta^{\ast} }(P^{\ast})(a|x)=1$ for all $x\in X$, and the final equality in the third line uses Lemma \ref{lem:AuxResult2}. To verify the result, it suffices to consider the last expression with $\lambda =\alpha $. Since $\Phi$ is a non-singular matrix and $\partial P_{\theta ^{\ast }}/\partial \alpha $ has full rank matrix, we conclude that the expression is square, symmetric, and negative definite, and, consequently, it must be non-singular.

Condition (f). By Young's theorem and Eq.\ \eqref{eq:SecondDeriv} with $\lambda =P$ and $(\alpha ,\theta _{f},P)=(\alpha ^{\ast },\theta _{f}^{\ast },P^{\ast })$,
\begin{align*}
& \frac{\partial^2 Q_{\infty }^{ML}(\alpha ^{\ast },\theta _{f}^{\ast },P^{\ast })}{\partial P\partial \alpha ^{\prime }} \\
&=\sum_{(a,x)\in A \times X}J^{\ast }(a,x)
\left[ \frac{-1}{( \Psi_{\theta ^{\ast }}(P^{\ast })(a,x)) ^{2}}\frac{\partial \Psi_{\theta ^{\ast }}(P^{\ast })(a,x)}{ \partial P}\frac{\partial \Psi_{\theta ^{\ast }}(P^{\ast })(a,x)}{ \partial \alpha ^{\prime }}+\frac{1}{\Psi_{\theta }(P)(a,x)}\frac{\partial^2 \Psi_{\theta ^{\ast }}(P^{\ast })(a,x)}{\partial P\partial \alpha ^{\prime } } \right] \\
&=\sum_{(a,x)\in A \times X}J^{\ast }(a,x)\left[ \frac{-1}{( \Psi_{\theta ^{\ast }}(P^{\ast })(a,x)) ^{2}}\frac{\partial \Psi_{\theta ^{\ast }}(P_{\theta ^{\ast }})(a,x)}{\partial P}\frac{\partial \Psi_{\theta ^{\ast }}(P^{\ast })(a,x)}{\partial \alpha ^{\prime }}+\frac{1}{\Psi_{\theta }(P)(a,x)}\frac{ \partial }{\partial \alpha ^{\prime }} \frac{\partial \Psi_{\theta ^{\ast }}(P_{\theta ^{\ast }})(a,x)}{\partial P} \right] \\
&= {\bf 0}_{|\tilde{A} \times X| \times d_{\alpha}}.
\end{align*}
where the second equality uses $P_{\theta ^{\ast }}=P^{\ast }$ and Young's theorem, and the last equality uses that the Jacobian matrix of $\Psi_{\theta ^{\ast }}$ with respect to $P$ is zero at $P_{\theta ^{\ast }}=P^{\ast }$.

\underline{Part 2:} Verify Assumption \ref{ass:EEpart3}.

Assumption \ref{ass:EEpart3}(b) holds as a corollary of Lemma \ref{lem:AsyDistBase2}. To verify Assumption \ref{ass:EEpart3}(a), consider the following argument. By direct computation,
\begin{align}
\frac{\partial Q_{n}^{ML}(\alpha ^{\ast },\theta _{f}^{\ast },P^{\ast })}{ \partial \alpha } &=\sum_{(a,x)\in A \times X}\hat{J}_{n}(a,x)\frac{1}{\Psi_{\theta ^{\ast }}(P^{\ast })(a,x)}\frac{\partial \Psi_{\theta ^{\ast }}(P^{\ast })(a,x)}{\partial \alpha } \notag\\
&=\frac{\partial\{ \{\ln P_{\theta ^{\ast }}(a|x)\}_{(a,x)\in A \times X}\}'}{\partial \alpha }\hat{J}_{n}= \frac{\partial P_{\theta ^{\ast }}'}{\partial \alpha } \Phi \Sigma\hat{J}_{n},\label{eq:QnML1}
\end{align}
where the second equality uses that $\partial \Psi_{\theta ^{\ast }}(P^{\ast })/\partial \alpha =\partial P_{\theta ^{\ast }}/\partial \alpha $ and the last equality uses Lemma \ref{lem:AuxResult2}. Also, by using an analogous argument applied to the population,
\begin{eqnarray}
\frac{\partial Q_{\infty }^{ML}(\alpha ^{\ast },\theta _{f}^{\ast },P^{\ast })}{\partial \alpha } 
&=& \frac{\partial P_{\theta ^{\ast }}'}{\partial \alpha } \Phi \Sigma {J}^* \notag\\
&=& \frac{\partial P_{\theta ^{\ast }}'}{\partial \alpha } \Phi \left\{  \left[ \left[ \mathbf{I}_{|\tilde{A} |\times |A|}-\{P^{\ast }(a|x)\}_{a\in \tilde{A}} \times \mathbf{1} _{1\times |A|} \right] / {m(x)}  \right]\times  \{ J^{\ast }(a ,x) \}_{a \in A} \right\}_{x\in X} \notag\\
&=& \frac{\partial P_{\theta ^{\ast }}'}{\partial \alpha } \Phi \left\{   \left[ \mathbf{I}_{|\tilde{A} |\times |A|}-\{P^{\ast }(a|x)\}_{a\in \tilde{A}} \times \mathbf{1} _{1\times |A|} \right]\times  \{ P^{\ast }(a |x) \}_{a \in A} \right\}_{x\in X} \notag \\
&=& \frac{\partial P_{\theta ^{\ast }}'}{\partial \alpha } \Phi {\bf 0}_{|\tilde{A}| \times |X| }  = {\bf 0}_{|\tilde{A}| \times |X| },\label{eq:QnML2}
\end{eqnarray}
where the third equality uses that $P^{\ast }(a|x)=J^{\ast }(a,x)/m^{\ast }(x)$ and the fourth equality uses that $\sum_{a \in A}P^{*}(a|x)=1$ for all $(a,x) \in A\times X$.

By combining Eqs.\ \eqref{eq:QnML1} and \eqref{eq:QnML2}, we conclude that:
\begin{equation*}
n^{\min \{ \delta ,1/2\} }\left[ 
\begin{array}{c}
\partial Q_{n}^{ML}(\alpha ^{\ast },\theta _{f}^{\ast },P^{\ast })/\partial \alpha \\
(\hat{\theta}_{f,n}-\theta _{f}^{\ast })
\end{array}
\right] =\left[ 
\begin{array}{cc}
\frac{\partial P_{\theta ^{\ast }}'}{\partial \alpha }\Phi \Sigma & \mathbf{0}_{|A \times X|\times d_{\theta_{f}}} \\ 
\mathbf{0}_{d_{\theta_{f}}\times |A \times X|} & \mathbf{I}_{d_{\theta_{f}} \times d_{\theta_{f}}}
\end{array}
\right] n^{\min \{ \delta ,1/2\} }\left(
\begin{array}{c}
\hat{J}_{n}-J^{\ast } \\ 
\hat{\theta}_{f,n}-\theta _{f}^{\ast }
\end{array}
\right).
\end{equation*}

From this and Lemma \ref{lem:AsyDistBase1}, we conclude that the desired result holds with:
\begin{align}
\zeta \sim \left[ 
\begin{array}{cc}
\frac{\partial P_{\theta ^{\ast }}'}{\partial \alpha }\Phi \Sigma & \mathbf{0}_{|A \times X|\times d_{\theta_{f}}} \\ 
\mathbf{0}_{d_{\theta_{f}}\times |A \times X|} & \mathbf{I}_{d_{\theta_{f}} \times d_{\theta_{f}}}
\end{array}
\right]  \Delta  N\left(B_{\Pi ^{\ast }}\times 1[\delta \leq 1/2], (diag(\Pi ^{\ast })-\Pi ^{\ast }\Pi ^{\ast \prime }) \times 1[\delta \geq 1/2]\right)
\label{eq:ZetaInMLE}
\end{align}

This completes the verification of Assumptions \ref{ass:EEpart2}-\ref{ass:EEpart3} and so Theorem \ref{thm:ANproof} applies. The specific formula for the asymptotic distribution relies on Eqs.\ \eqref{eq:SecondDerivMLE} and \eqref{eq:ZetaInMLE}.
\end{proof}

%%%%%%%% DIVIDER %%%%%%%%%%%%

\begin{proof}[Proof of Theorem \ref{thm:MD_main}] 
	This result is a corollary of Theorem \ref{thm:ANproof}. To complete the proof, we need to verify Assumptions \ref{ass:EEpart2}-\ref{ass:EEpart3}. We anticipate that $Q_{\infty }^{MD}(\theta ,P)=-[P^{\ast }-\Psi_{\theta }(P)]^{\prime }W^{\ast }[P^{\ast }-\Psi_{\theta }(P)]$.

\underline{Part 1:} Verify the conditions in Assumption \ref{ass:EEpart2}.

Condition (a). First, we show that $\sup_{(\theta ,P)\in \Theta \times \Theta _{P}}|Q_{n}^{MD}(\theta ,P)-Q_{\infty }^{MD}(\theta ,P)|=o_{p_{n}}(1)$. Consider the following argument:
\begin{eqnarray*}
\sup_{(\theta ,P)\in \Theta \times \Theta _{P}}|Q_{n}^{MD}(\theta ,P)-Q_{\infty }^{MD}(\theta ,P)| &=&\sup_{(\theta ,P)\in \Theta \times \Theta _{P}}\left\vert
\begin{array}{c}
-(\hat{P}_{n}-P^{\ast })^{\prime }\hat{W}_{n}[\hat{P}_{n}-\Psi_{\theta }(P)] \\
-(P^{\ast }-\Psi_{\theta }(P))^{\prime }[\hat{W}_{n}-W^{\ast }][\hat{P} _{n}-\Psi_{\theta }(P)] \\
-(P^{\ast }-\Psi_{\theta }(P))^{\prime }W^{\ast }(\hat{P}_{n}-P^{\ast })
\end{array}
\right\vert \\
&\leq &\Vert \hat{P}_{n}-P^{\ast }\Vert \times (\Vert \hat{W}_{n}-W^{\ast }\Vert +2\Vert W^{\ast }\Vert) +\Vert \hat{W}_{n}-W^{\ast }\Vert
\end{eqnarray*}
Second, since $\Psi_{\theta }(P)(a|x):\Theta \times \Theta _{P}\to \mathbb{R}$ is continuous in $(\theta ,P)$ for all $(a,x)$, $ Q_{\infty }^{MD}(\theta ,P):\Theta \times \Theta _{P}\to \mathbb{R}$ is continuous in $(\theta ,P)$. In turn, since $\Theta \times \Theta _{P}$ is compact, $Q_{\infty }^{MD}(\theta ,P):\Theta \times \Theta _{P}\to \mathbb{R}$ is uniformly continuous in $(\theta ,P)$. Third, $(\hat{\theta}_{f,n},\tilde{P}_{n})-(\theta _{f}^{\ast },P^{\ast })=o_{p_{n}}(1)$, where $\tilde{P}_{n}$ is the arbitrary sequence in condition (a). By combining these with \citet[Lemma 24.1]{gourieroux/monfort:1995b}, the result follows.

Condition (b). $Q_{\infty }^{MD}(\alpha ,\theta _{f}^{\ast },P^{\ast })=- [ P^{\ast }-\Psi_{(\alpha ,\theta _{f}^{\ast })}(P^{\ast })] ^{\prime }W^{\ast }[ P^{\ast }-\Psi_{(\alpha ,\theta _{f}^{\ast })}(P^{\ast })] $ is uniquely maximized at $\alpha ^{\ast }$. First, notice that $\Psi_{(\alpha ^{\ast },\theta _{f}^{\ast })}(P^{\ast })=P^{\ast }$ and so $Q_{\infty }^{MD}(\alpha ^{\ast },\theta _{f}^{\ast },P^{\ast })=0$. Second, consider any $\tilde{\alpha}\in \Theta _{\alpha }\backslash \alpha ^{\ast }$. By the identification assumption, $\Psi_{( \tilde{\alpha},\theta _{f}^{\ast })}(P^{\ast })\not=\Psi_{(\alpha ^{\ast },\theta _{f}^{\ast })}(P^{\ast })=P^{\ast }$. Since $W^{\ast }$ is positive definite, $Q_{\infty }^{MD}( \tilde{\alpha},\theta _{f}^{\ast },P^{\ast })>0$.

Condition (c). This result follows from the fact that $\Psi_{\theta }(P)(a|x):\Theta \times \Theta _{P}\to \mathbb{R} $ is twice continuously differentiable in $(\theta ,P)$ for all $(a,x) \in A \times X$.

Condition (d). By the same argument as in the verification of condition (c), $Q_{\infty }^{MD}(\theta ,P):\Theta \times \Theta _{P}\to \mathbb{R}$ is twice continuously differentiable in $(\theta ,P)$. Since $\Psi_{\theta }(P)(a|x)$ is twice continuously differentiable in $(\theta ,P)$ for all $(a,x) \in \tilde{A} \times X$, we conclude that $\partial \Psi_{\theta }(P)(a,x)/\partial \lambda $ and $\partial \Psi_{\theta }(P)(a,x)/\partial \alpha\partial \lambda' $ are continuous in $(\theta ,P)$ for all $ \lambda \in \{\theta ,P\}$ and $(a,x) \in \tilde{A} \times X$. From this and the fact that $ \Theta \times \Theta _{P}$ is compact, $\max_{(a,x)\in A \times X}\sup_{(\theta ,P)\in \Theta \times \Theta _{P}}\Vert \partial \Psi_{\theta }(P)(a,x)/\partial \lambda \Vert <\infty $ and $\max_{(a,x)\in A \times X}\sup_{(\theta ,P)\in \Theta \times \Theta _{P}}\Vert \partial \Psi_{\theta }(P)(a,x)/\partial \alpha \partial \lambda'\Vert <\infty $. From this, $\hat{P}_{n}-P^{\ast }=o_{p_{n}}(1)$, and $\hat{W}_{n}-W^{\ast }=o_{p_{n}}(1)$, the desired result follows.

Condition (e). For any $\lambda \in \{\alpha ,\theta _{f},P\}$, direct computation shows that:
\begin{equation}
	\frac{\partial^2 Q_{\infty }^{MD}(\alpha ,\theta _{f},P)}{\partial \lambda \partial \alpha'}
	=2 \left[ \frac{\partial }{\partial \alpha' }\frac{ \partial \Psi_{\theta }(P)'}{\partial \lambda}W^{\ast } (P^{\ast }-\Psi_{\theta }(P)) -\frac{\partial \Psi_{\theta }(P)'}{\partial \lambda }W^{\ast } \frac{ \partial \Psi_{\theta }(P)}{\partial \alpha' } \right].
\label{eq:SecondDerivMD}
\end{equation}
This function is continuous and if we evaluate at $(\alpha ,\theta _{f},P)=(\alpha ^{\ast },\theta _{f}^{\ast },P^{\ast })$, we obtain: 
\begin{equation*} 
	\frac{\partial Q_{\infty }^{MD}(\alpha ^{\ast },\theta _{f}^{\ast },P^{\ast })}{\partial \lambda \partial \alpha ^{\prime }}
	= -2\frac{\partial \Psi_{\theta^* }(P^*)'}{\partial \lambda }W^{\ast } \frac{ \partial \Psi_{\theta^* }(P^*)}{\partial \alpha' }
	= -2\frac{\partial P _{\theta^* }'}{\partial \lambda }W^{\ast } \frac{ \partial P_{\theta^* }}{\partial \alpha' }
\end{equation*}
where the first line uses that $P^{\ast }=\Psi_{\theta ^{\ast }}(P^{\ast })$ and $\partial \Psi_{\theta ^{\ast }}(P^{\ast })/\partial \alpha =\partial P_{\theta ^{\ast }}/\partial \alpha $. To verify the result, it suffices to consider the last expression with $\lambda =\alpha $. By assumption, this expression is square, symmetric, and negative definite, and, consequently, it must be non-singular.

Condition (f).  By Young's theorem and Eq.\ \eqref{eq:SecondDerivMD} with $\lambda =P$ and $(\alpha ,\theta _{f},P)=(\alpha ^{\ast },\theta _{f}^{\ast },P^{\ast })$,
\begin{equation*}
\frac{\partial^2 Q_{\infty }^{MD}(\alpha ^{\ast },\theta _{f}^{\ast },P^{\ast })}{\partial P\partial \alpha ^{\prime }}
=-2\frac{\partial \Psi_{\theta ^{\ast }}(P^{\ast })'}{\partial P }W^{\ast }\frac{\partial \Psi_{\theta ^{\ast }}(P^{\ast })}{\partial \alpha '}
=- 2 \frac{ \partial \Psi_{\theta ^{\ast }}(P_{\theta^{\ast }})'}{\partial P }W^{\ast }\frac{\partial \Psi_{\theta ^{\ast }}(P_{\theta^{\ast }})}{\partial \alpha'}={\bf 0}_{|\tilde{A} \times X| \times d_{\alpha}}.
\end{equation*}
where the last equality uses that the Jacobian matrix of $\Psi_{\theta ^{\ast }}$ with respect to $P$ is zero at $P_{\theta ^{\ast }}=P^{\ast }$.

\underline{Part 2:} Verify the conditions in Assumption \ref{ass:EEpart3}.

Assumption \ref{ass:EEpart3}(b) holds as a corollary of Lemma \ref{lem:AsyDistBase2}. To verify Assumption \ref{ass:EEpart3}(a), consider the following argument. By direct computation,
\begin{equation*}
\frac{\partial Q_{n}^{MD}(\alpha ^{\ast },\theta _{f}^{\ast },P^{\ast })}{ \partial \alpha }
=2\frac{\partial \Psi_{\theta ^{\ast }}(P^{\ast })'}{\partial \alpha }\hat{W}_{n}[\hat{P}_{n}-\Psi_{\theta ^{\ast }}(P^{\ast })]
=2\frac{\partial P_{\theta^* }'}{\partial \alpha }W^{\ast }[\hat{P}_{n}-P^{\ast }]+o_{p_{n}}(1),
\end{equation*}
where the last equality uses that $\Psi_{\theta ^{\ast }}(P^{\ast })=P^{\ast }$, $\partial \Psi_{\theta ^{\ast }}(P^{\ast })'/\partial \alpha =\partial P_{\theta ^{\ast }}'/\partial \alpha $, $\hat{P}_{n}-P^{\ast }=o_{p_{n}}(1)$, and $\hat{W}_{n}-W^{\ast }=o_{p_{n}}(1)$. We then conclude that:
\begin{equation*}
n^{\min \{ \delta ,1/2\} }\left[
\begin{array}{c}
\partial Q_{n}^{MD}(\alpha ^{\ast },\theta _{f}^{\ast },P^{\ast })/\partial \alpha \\
(\hat{\theta}_{f,n}-\theta _{f}^{\ast })
\end{array}
\right] =\left[ 
\begin{array}{cc}
2\frac{\partial P_{\theta^{\ast } }'}{\partial \alpha }W^{\ast } & \mathbf{0} _{|A \times X|\times d_{\theta_{f}}} \\
\mathbf{0}_{d_{\theta _{f}}\times |A \times X|} & \mathbf{I}_{d_{\theta _{f}} \times d_{\theta _{f}}}
\end{array}
\right] n^{\min \{ \delta ,1/2\} }\left[ 
\begin{array}{c}
(\hat{P}_{n}-P^{\ast }) \\ 
(\hat{\theta}_{f,n}-\theta _{f}^{\ast })
\end{array}
\right] +o_{p_{n}}(1).
\end{equation*}

From this and Lemma \ref{lem:AsyDistBase2}, we conclude that the desired result holds with:
\begin{align}
\zeta \sim \left[ 
\begin{array}{cc}
\frac{\partial P_{\theta ^{\ast }}'}{\partial \alpha }W^* \Sigma & \mathbf{0}_{|A \times X|\times d_{\theta_{f}}} \\ 
\mathbf{0}_{d_{\theta_{f}}\times |A \times X|} & \mathbf{I}_{d_{\theta_{f}} \times d_{\theta_{f}}}
\end{array}
\right]  \Delta  N\left(B_{\Pi ^{\ast }}\times 1[\delta \leq 1/2], (diag(\Pi ^{\ast })-\Pi ^{\ast }\Pi ^{\ast \prime }) \times 1[\delta \geq 1/2]\right)
\label{eq:ZetaInMD}
\end{align}

This completes the verification of Assumptions \ref{ass:EEpart2}-\ref{ass:EEpart3} and so Theorem \ref{thm:ANproof} applies. The specific formula for the asymptotic distribution relies on Eqs.\ \eqref{eq:SecondDerivMD} and \eqref{eq:ZetaInMD}.
\end{proof} 

%%%%%%%% DIVIDER %%%%%%%%%%%%

\subsection{Proofs of lemmas}

%%%%%%%% DIVIDER %%%%%%%%%%%%

\begin{proof}[Proof of Lemma \ref{lem:PolicyProperties}]
	This proof follows from \citet[Propositions 1-2]{aguirregabiria/mira:2002}.
\end{proof}

%%%%%%%% DIVIDER %%%%%%%%%%%%

\begin{proof}[Proof of Lemma \ref{lem:AuxResultsOnCCP}]
	Parts (a)-(b) follow from \citet[Pages 1015-6]{rust:1988}. Part (c) follows from combining Lemma \ref{lem:PolicyProperties} and Assumption \ref{ass:Identification}.
\end{proof}

%%%%%%%% DIVIDER %%%%%%%%%%%%

\begin{lemma}\label{lem:AsyDistBase1} 
Assume Assumptions \ref{ass:LocalMiss}-\ref{ass:iid}. Then,
\begin{equation}
n^{\min \{\delta ,1/2\}}\left(
\begin{array}{c}
\hat{J}_{n}-J^{\ast } \\ 
\hat{\theta}_{f,n}-\theta _{f}^{\ast }
\end{array}
\right) \overset{d}{\to }\Delta \times N\left(  B_{\Pi ^{\ast }}\times 1[\delta \leq 1/2],(diag(\Pi ^{\ast })-\Pi ^{\ast }\Pi ^{\ast \prime }) \times 1[\delta \geq 1/2]\right) , \label{eq:AsyDistBase}
\end{equation}
with $\Delta$ as in Eq.\ \eqref{eq:Delta}.
\end{lemma}
%%%%%%% DIVIDER %%%%%%%%%%%%
\begin{proof}
Under Assumption \ref{ass:iid}, the triangular array CLT (e.g.\ \citet[page 369]{davidson:1994}) implies that:
\[
\sqrt{n}(\hat{\Pi}_{n}-\Pi _{n}^{\ast })\overset{d}{\rightarrow } N(0,diag(\Pi ^{\ast })-\Pi ^{\ast }\Pi ^{\ast \prime }).
\]
If we combine this with Assumption \ref{ass:LocalMiss},
\begin{equation}
n^{\min \{\delta ,1/2\}}(\hat{\Pi}_{n}-\Pi ^{\ast })\overset{d}{\rightarrow } N(B_{\Pi ^{\ast }}\times 1[\delta \leq 1/2],(diag(\Pi ^{\ast })-\Pi ^{\ast }\Pi ^{\ast \prime })\times 1[\delta \geq 1/2]). \label{eq:CLT}
\end{equation}

Also, notice that:
\[
n^{\min \{\delta ,1/2\}}\left( 
\begin{array}{c}
\hat{J}_{n}-J^{\ast } \\ 
\hat{\theta}_{f,n}-\theta _{f}^{\ast }
\end{array}
\right) =n^{\min \{\delta ,1/2\}}(F(\hat{\Pi}_{n})-F(\Pi ^{\ast })),
\]
where $F:\mathbb{R}^{|A\times X\times X|}\rightarrow \mathbb{R}^{|A\times X|+d_{\theta _{f}}}$ is defined as follows. For coordinates $j \leq |A\times X|$ where $j$ represents the corresponding coordinate $(a,x)\in A\times X$, $F_{j}(z)\equiv \sum_{\tilde{x}^{\prime }\in X}z_{(a,x,\tilde{x} ^{\prime })}$, and for coordinates $j>|A\times X|$, $F_{j}(z)\equiv G_{1,j}(z)$. By definition of $G_{1}$, $\hat{\theta}_{f,n}=G_{1}(\hat{\Pi}_{n})$, and by Assumption \ref{ass:Preliminary}, $\theta _{f}^{\ast }=G_{1}(\Pi ^{\ast })$ and $F$ is continuously differentiable at $\Pi ^{\ast }$. By direct computation, $\Delta =\partial{F(\Pi ^{\ast })}/\partial{\Pi'}$. Then, the result follows from the delta method and Eq.\ \eqref{eq:CLT}.
\end{proof}

%%%%%%%% DIVIDER %%%%%%%%%%%%

\begin{lemma} \label{lem:AsyDistBase2}
	Assume Assumptions \ref{ass:LocalMiss}-\ref{ass:iid}. Then,
	\begin{align}
	&n^{\min \{\delta ,1/2\}}\left(
	\begin{array}{c}
	\hat{P}_{n}-P^{\ast } \\ 
	\hat{\theta}_{f,n}-\theta _{f}^{\ast }
	\end{array}
	\right)\notag \\
	&\overset{d}{\to }\left[ 
\begin{array}{cc}
\Sigma & \mathbf{0}_{|\tilde{A} \times X|\times d_{\theta _{f}}} \\
\mathbf{0}_{d_{\theta _{f}}\times |\tilde{A} \times X|} & \mathbf{I}_{d_{\theta _{f}}\times d_{\theta _{f}}}
\end{array}
\right] \Delta \times N\left(  B_{\Pi ^{\ast }} 1[\delta \leq 1/2],(diag(\Pi ^{\ast })-\Pi ^{\ast }\Pi ^{\ast \prime }) 1[\delta \geq 1/2]\right) , \label{eq:AsyDistBaseP}
	\end{align}
	with $\Delta$ as in Eq.\ \eqref{eq:Delta} and $\Sigma$ as in Eq.\ \eqref{eq:Sigma}.
\end{lemma}
%%%%%%%% DIVIDER %%%%%%%%%%%%
\begin{proof}
Let $F:\mathbb{R}^{|A \times X|+d_{\theta _{f}}}\to \mathbb{R}^{|\tilde{A} \times X|+d_{\theta _{f}}}$ be defined as follows. For coordinates $j \leq | \tilde{A} \times X|$ with $j$ representing coordinate $(a,x)\in \tilde{A} \times X$, $F_{j}(z)\equiv z_{(a,x)}/\sum_{a\in A}z_{(\tilde{a},x)}$, and for $j>|\tilde{A} \times X|$, $F_{j}(z)=z_{j}$. Notice that $F((\hat{J}_{n}',\hat{\theta}_{f,n}')')\equiv (\hat{P}_{n}',\hat{ \theta}_{f,n}')'$ and $F(({J^*}',{\theta _{f}^*}')')\equiv ({P^*}',{\theta _{f}^*}')'$ by definition of $F$. If we verify $F$ is continuously differentiable at $z=({J^*}',{\theta _{f}^*}')'$ and
\begin{equation}
	\frac{\partial F(({J^*}',{\theta _{f}^*}')')}{\partial z'}~=~\left[ 
\begin{array}{cc}
\Sigma & \mathbf{0}_{|\tilde{A} \times X|\times d_{\theta _{f}}} \\
\mathbf{0}_{d_{\theta _{f}}\times |\tilde{A} \times X|} & \mathbf{I}_{d_{\theta _{f}}\times d_{\theta _{f}}}
\end{array}
\right], \label{eq:GradientF}
\end{equation}
then the result follows from the delta method and Lemma \ref{lem:AsyDistBase2}. We do this next. 

Consider $(j,\check{j}) \in \{1,\dots,|\tilde{A} \times X|\} \times \{1,\dots,|A \times X|\}$ representing $(a,x) \in \tilde{A} \times X$ and $(\check{a},\check{x}) \in A \times X$. For any $j>|\tilde{A} \times X|$, $F_{j} $ is continuously differentiable and ${\partial F_{j}(z)}/{\partial z_{\check{j}}} = 1[j=\check{j}]$. For any $j \leq |\tilde{A} \times X|$,
\begin{eqnarray*}
\frac{\partial F_{j}(z)}{\partial z_{\check{j}}} &=& 1[x= \check{x}]
\left[ \left(\frac{\sum_{\tilde{a}\in A}z_{(\tilde{a},x)}-z_{(\check{a}, x)}}{(\sum_{\tilde{a}\in A}z_{(\tilde{a},x)})^{2}}\right) 1[a=\check{a}]
+ \left(\frac{-z_{(\check{a},x)}}{(\sum_{\tilde{a}\in A}z_{(\tilde{a}, x)})^{2}}\right) 1[a\not=\check{a}] \right] \\
&=&\frac{1[x=\check{x}]}{(\sum_{\tilde{a}\in A}z_{(\tilde{a},{x})})} \left[1[a=\check{a}]-\frac{z_{(\check{a},x)}}{(\sum_{\tilde{a}\in A}z_{( \tilde{a},x)})}\right],
\end{eqnarray*}
provided that $\sum_{\tilde{a}\in A}z_{(\tilde{a},x)}>0$. Since $ \sum_{\tilde{a}\in A}J^{\ast }(\tilde{a},x)>0$ for all $x\in X$, $F$ is continuously differentiable at $(({J^*}',{\theta _{f}^*}')')$. By combining the formula for the derivatives from all coordinates, Eq.\ \eqref{eq:GradientF} follows.
\end{proof}

%%%%%%%% DIVIDER %%%%%%%%%%%%

\begin{lemma}\label{lem:AuxResult2}  
For any $\lambda ,\tilde{\lambda}\in \{\theta _{f},\alpha \}$, the following algebraic results hold:
\begin{eqnarray*}
\frac{\partial \{\ln P_{\theta ^{\ast }}(a|x)\}_{(a,x)\in A \times X}'}{\partial \lambda } &=&\frac{\partial P_{\theta ^{\ast }}'}{\partial \lambda } \Phi  \Sigma  \\
\sum_{(a,x)\in A \times X}J^{\ast }(a,x)\frac{\partial \ln P_{\theta ^{\ast }}(a|x)}{\partial \lambda}\frac{\partial \ln P_{\theta ^{\ast }}(a|x)}{ \partial \tilde{\lambda} ^{\prime }}
&=&\frac{\partial P_{\theta ^{\ast }}'}{ \partial \lambda } \Phi \frac{\partial P_{\theta ^{\ast }}}{\partial \tilde{\lambda}'},
\end{eqnarray*}
with $\Phi$ and $\Sigma$ as in Eq.\ \eqref{eq:Sigma}.
\end{lemma}
%%%%%%%% DIVIDER %%%%%%%%%%%%
\begin{proof}
Before deriving the results, consider some preliminary observations. For any $\lambda \in \{\alpha,\theta_f,P\}$, $\sum_{a\in A}P_{\theta ^{\ast }}(a|x)=1$ and so $\partial P_{\theta ^{\ast }}(|A||x)/\partial \lambda=-\sum_{a\in \tilde{A}}\partial P_{\theta ^{\ast }}(a|x)/\partial \lambda $. Also, for any $\lambda \in \{\alpha,\theta_f,P\}$ and $(a,x)\in A \times X$, $ P^{\ast }(a|x)=P_{\theta ^{\ast }}(a|x)$ and so $(\partial P_{\theta ^{\ast }}(a|x)/\partial \lambda )(1/P^{\ast }(a|x))=\partial \ln P_{\theta ^{\ast }}(a|x)/\partial \lambda $.

For the first result, consider the following derivation: 
\begin{eqnarray*}
	\frac{\partial P_{\theta ^{\ast }}'}{\partial \lambda } \Phi \Sigma &=& \frac{\partial \{P_{\theta ^{\ast }}(a|x)\}_{(a,x)\in \tilde{A} \times X}'}{\partial \lambda } \times \{diag\{ \Phi_x \Sigma_x \}_{x \in X}\} \\
	&=& \frac{\partial \{P_{\theta ^{\ast }}(a|x)\}_{(a,x)\in \tilde{A} \times X}'}{\partial \lambda } \times diag\{[diag\{\{1/P^{\ast }(a|x)\}_{a\in \tilde{A}}\},(-1/P^{\ast }(|A||x))\mathbf{1}_{|\tilde{A} |\times 1}]\}_{x\in X}\\
	&=& \frac{ \partial \{ \{ \ln P_\theta^* (a|x) \}_{(a,x) \in A \times X} \}' }{\partial \lambda},
\end{eqnarray*}
where the last equality uses the preliminary observations. 

For the second result, consider the following derivation:
\begin{eqnarray*}
	\frac{\partial P_{\theta ^{\ast } }^{\prime }}{\partial \lambda }\Phi \frac{ \partial P_{\theta ^{\ast }}}{\partial \tilde{\lambda}^{\prime }}&=&\frac{ \partial \{P_{\theta ^{\ast }}(a|x)\}_{(a,x)\in \tilde{A}\times X}^{\prime } }{\partial \lambda }\times \Phi \times \frac{\partial \{P_{\theta ^{\ast }}(a|x)\}_{(a,x)\in \tilde{A}\times X}}{\partial \tilde{\lambda}^{\prime }} \\
&=&\left\{ 
\begin{array}{c} 
\frac{\partial \{P_{\theta ^{\ast }}(a|x)\}_{(a,x)\in \tilde{A}\times X}^{\prime }}{\partial \lambda }diag\left\{ m(x)\left[ diag\left\{ \{{1}/{P^{\ast }(a|x)}\}_{a\in \tilde{A} }\right\} \right] \right\} _{x\in X}\frac{\partial \{P_{\theta ^{\ast }}(a|x)\}_{(a,x)\in \tilde{A}\times X}}{\partial \tilde{\lambda}^{\prime }}+ \\
\frac{\partial \{P_{\theta ^{\ast }}(a|x)\}_{(a,x)\in \tilde{A}\times X}^{\prime }}{\partial \lambda }diag\left\{ m(x)\left[ \mathbf{1}_{|\tilde{A} |\times |\tilde{A}|}/({1-\sum\nolimits_{a\in \tilde{A}}P^{\ast }(a|x)}) \right] \right\} _{x\in X}\frac{\partial \{P_{\theta ^{\ast }}(a|x)\}_{(a,x)\in \tilde{A}\times X}}{\partial \tilde{\lambda}^{\prime }}
\end{array}
\right\}  \\
&=&\sum_{(a,x)\in \tilde{A}\times X}m(x)\frac{\partial \ln P_{\theta ^{\ast }}(a|x)}{\partial \lambda }\frac{\partial P_{\theta ^{\ast }}(a|x)}{\partial \tilde{\lambda}^{\prime }}+\sum_{x\in X} \frac{m(x)}{P\left( \left\vert A\right\vert |x\right) }\frac{\sum_{a\in \tilde{A}}\partial P_{\theta ^{\ast }}(a|x)}{\partial \lambda }\frac{\sum_{\tilde{a}\in \tilde{A }}\partial P_{\theta ^{\ast }}(\tilde{a}|x)}{\partial \lambda ^{\prime }} \\
&=&\sum_{(a,x)\in A\times X}J^{\ast }(a,x)\frac{\partial \ln P_{\theta ^{\ast }}(a|x)}{\partial \lambda }\frac{\partial \ln P_{\theta ^{\ast }}(a|x) }{\partial \tilde{\lambda}^{\prime }},
\end{eqnarray*}
where the last equality uses the preliminary observations.
\end{proof}

\subsection{Review of results on extremum estimators}

The purpose of this section is to state well-known results regarding the consistency and asymptotic normality of extremum estimators under certain regularity conditions. These results are referenced in our formal arguments. Relative to the standard versions in the literature (e.g.\ \cite{mcfadden/newey:1994}), our results allow for: (a) a rate of convergence that may differ from $\sqrt{n}$ and (b) a sequence of DGPs that may change with sample size. Both of these features are important for our theoretical results. We omit the proofs for reasons of brevity but theses are available from the authors upon request.

\begin{theorem} \label{thm:consistencyEE}  
	Assume the following:
	\begin{enumerate}[(a)]
	\item $Q_{n}(\theta )$ converges uniformly in probability to $Q(\theta )$ along $\{p_{n}\}_{n\geq 1}$.
	\item $Q(\theta )$ is upper semi-continuous, i.e., for any $\{\theta _{n}\}_{n\geq 1}$ with $\theta _{n}\to \tilde{\theta}$, $\lim \sup Q(\theta _{n})\leq Q(\tilde{\theta})$.
	\item $Q(\theta )$ is uniquely maximized at $\theta =\theta ^{\ast }$.
	\end{enumerate}
	Then, $\hat{\theta}_{n}={\arg\max}_{\theta \in \Theta }Q_{n}(\theta )$ satisfies $\hat{\theta}_{n}=\theta ^{\ast }+o_{p_{n}}(1)$.
\end{theorem}

%%%%%%%% DIVIDER %%%%%%%%%%%%

\begin{theorem}\label{thm:ANforEE} 
Consider an estimator $\hat{\theta}_{n}$ of a parameter $ \theta ^{\ast }$ s.t.\ $\hat{\theta}_{n}={\arg\max}_{\theta \in \Theta }Q_{n}(\theta )$. Furthermore,
\begin{enumerate}[(a)]
\item $\hat{\theta}_{n}=\theta ^{\ast }+o_{p_{n}}(1)$,
\item $\theta ^{\ast }$ belongs to the interior of $\Theta $,
\item $Q_{n}$ is twice continuously differentiable in a neighborhood $ \mathcal{N}$ of $\theta ^{\ast }$ w.p.a.1,
\item For some $\delta>0$, ${n}^{\delta}\partial Q_{n}(\theta ^{\ast })/\partial \theta \overset{d}{ \to }Z$ for some random variable $Z$ along $\{p_{n}\}_{n\geq 1}$,
\item $\sup_{\theta \in \mathcal{N}}||\partial ^{2}Q_{n}(\theta )/\partial \theta \partial \theta^{\prime }-H(\theta )||=o_{p_{n}}(1)$ for some function $H:\mathcal{N}\to \mathbb{R}^{k\times k}$ that is continuous at $ \theta ^{\ast }$,
\item $H(\theta ^{\ast })$ is non-singular.
\end{enumerate}
Then, ${n}^{\delta}(\hat{\theta}_{n}-\theta ^{\ast })=-H(\theta ^{\ast })^{-1} {n}^{\delta}\partial Q_{n}(\theta ^{\ast })/\partial \theta +o_{p_{n}}(1) \overset{d}{\to } -H(\theta ^{\ast })^{-1}Z$ along $\{p_{n}\}_{n\geq 1}$.
\end{theorem}

%%%%%%%%% DIVIDER %%%%%%%%%%%%
}

% \newpage

%\bibliographystyle{plain}
\bibliography{BIBLIOGRAPHY}

\section{Supplemental materials}

\subsection{Additional simulation results for the first design}

Table \ref{tab:miss3a} is the last set of results for the Monte Carlo design described in Section \ref{sec:MonteCarlos}. This table presents results under asymptotically overwhelming local misspecification (i.e.\ $\delta=1/3$), with estimators scaled by the regular $\sqrt{n}$-rate. According to our theoretical results, the presence of overwhelming local misspecification implies that these estimators no longer converge at the regular $\sqrt{n}$-rate, but rather at the $n^{1/3}$-rate. In accordance with this prediction, Table \ref{tab:miss3a} reveals that the asymptotic bias of the estimators does not appear to converge when scaled by the regular $\sqrt{n}$-rate.

\begin{table}[H]
\begin{center}
\scalebox{0.85}
{
\begin{tabular}{cc|ccc|ccc|ccc}\hline
\hline
\multicolumn{1}{c}{\multirow{2}[4]{*}{$K$}} & \multicolumn{1}{c|}{\multirow{2}[4]{*}{Statistic}} & \multicolumn{3}{c|}{$K$-MD(${\bf I}_{|\tilde{A} \times X| \times |\tilde{A} \times X|}$)} & \multicolumn{3}{c|}{$K$-MD($W_{AV}^{*}$)} & \multicolumn{3}{c}{$K$-ML}\\
\multicolumn{1}{c}{} & \multicolumn{1}{c|}{} &
\multicolumn{1}{c}{$n=200$} & \multicolumn{1}{c}{$n=500$} & \multicolumn{1}{c|}{$n=1,000$} &
\multicolumn{1}{c}{$n=200$} & \multicolumn{1}{c}{$n=500$} & \multicolumn{1}{c|}{$n=1,000$} &
\multicolumn{1}{c}{$n=200$} & \multicolumn{1}{c}{$n=500$} & \multicolumn{1}{c}{$n=1,000$}
\\
\hline
  & $\sqrt{n}~$Bias & 0.98 & 1.14 & 1.30 & 1.04 & 1.23 & 1.42 & 1.10 & 1.29 & 1.47 \\
$1$ & $\sqrt{n}~$SD & 
0.34 & 0.33 & 0.31 & 0.30 & 0.29 & 0.28 & 0.29 & 0.28 & 0.28 \\
  & $n~$MSE & 
1.08 & 1.40 & 1.80 & 1.17 & 1.60 & 2.08 & 1.28 & 1.74 & 2.22 \\
\hline
  & $\sqrt{n}~$Bias & 
0.90 & 1.12 & 1.30 & 0.97 & 1.22 & 1.41 & 1.05 & 1.27 & 1.46 \\
$2$ & $\sqrt{n}~$SD & 
0.34 & 0.33 & 0.31 & 0.29 & 0.29 & 0.28 & 0.28 & 0.28 & 0.27 \\
  & $n~$MSE & 
0.92 & 1.37 & 1.79 & 1.03 & 1.56 & 2.07 & 1.17 & 1.70 & 2.20 \\
\hline
  & $\sqrt{n}~$Bias & 
0.90 & 1.12 & 1.30 & 0.97 & 1.22 & 1.41 & 1.05 & 1.27 & 1.46 \\
$3$ & $\sqrt{n}~$SD & 
0.34 & 0.33 & 0.31 & 0.29 & 0.29 & 0.28 & 0.28 & 0.28 & 0.27 \\
  & $n~$MSE & 
0.92 & 1.37 & 1.79 & 1.03 & 1.57 & 2.07 & 1.17 & 1.70 & 2.20 \\
\hline
  & $\sqrt{n}~$Bias & 
0.90 & 1.12 & 1.30 & 0.97 & 1.22 & 1.41 & 1.05 & 1.27 & 1.46 \\
$10$ & $\sqrt{n}~$SD & 
0.34 & 0.33 & 0.31 & 0.29 & 0.29 & 0.28 & 0.28 & 0.28 & 0.27 \\
  & $n~$MSE &
0.92 & 1.37 & 1.79 & 1.03 & 1.56 & 2.07 & 1.17 & 1.70 & 2.20 \\
\hline
\hline
\end{tabular}
	}\end{center}
		\caption{Simulation results under local misspecification with $\tau_{n} \propto n^{-1/3}$ and using the regular scaling (i.e.\ $\sqrt{n}$).}
	\label{tab:miss3a}
\end{table}

\subsection{Simulation results for the second misspecification design}

This section describes Monte Carlo simulation results for a second misspecification design. The econometric model is exactly as the one described in Section \ref{sec:MonteCarlos}. The true DGP is analogous to the second illustration in Section \ref{sec:Misspecification}, and it is inspired by the presence of unobserved heterogeneity along the lines of \cite{arcidiacono/miller:2011}. 

We simulate data composed of two types of agents, A and B. Both types of agents behave exactly according to the model and only differ in the parameter value of their utility functions. Recall from Section \ref{sec:MonteCarlos} that the utility function is specified as follows:
\begin{equation*}
u_{\theta _{u}}(x,a)~=~-\theta _{u,1}\times  1[ a=2]~-\theta _{u,2}\times  1[ a=1]x,
%\label{eq:utilityMC}
\end{equation*}
Agents of type A have $(\theta _{u,1},\theta _{u,2} ) = (1,0.05)$, while agents of type B have $(\theta _{u,1},\theta _{u,2} ) = (0,95,-0.05)$. 

The econometric model is then misspecified in the sense that it presumes a homogenous sample. We use $\tau_{n} \in (0,1)$ to denote the proportion of agents of type B in the population. We impose local misspecification by setting $\tau_n \equiv n^{-\delta}$ with $\delta \in \{1/3,1/2,1\}$. The rest of the parameters used to implement the Monte Carlo simulation are exactly as in Section \ref{sec:MonteCarlos}.

For the sake of comparison with the first simulation design, we present results for the estimator of $\theta_{u,2}$. The simulation results are qualitatively similar to the ones obtained in the previous section and support all of our theoretical conclusions. In particular, the results in Tables \ref{tab:miss1_2},  \ref{tab:miss2_2}, \ref{tab:miss3b_2}, and \ref{tab:miss3a_2} are analogous to Tables \ref{tab:miss1},  \ref{tab:miss2}, \ref{tab:miss3b}, and \ref{tab:miss3a}, respectively. We refer to Section \ref{sec:MonteCarlos} for a description of these results.

% Table \ref{tab:miss1_2} provides results under asymptotically irrelevant local misspecification, i.e., $\delta=1$. As expected, these are virtually identical to the ones obtained under correct specification (see Table \ref{tab:miss0}). Table \ref{tab:miss2_2} provides results under local misspecification that vanishes at the knife-edge rate, i.e., $\delta=1/2$. The results show that the biases are non-zero but independent of the number of iterations $K$. Tables \ref{tab:miss3a_2} and \ref{tab:miss3b_2} provide results under asymptotically overwhelming local misspecification, i.e., $\delta=1/3$.  Table \ref{tab:miss3a_2} uses the regular $\sqrt{n}$-rate and, as expected, the asymptotic bias diverges. Table \ref{tab:miss3b_2} uses the appropriate $n^{1/3}$-rate. Once we use this rate, the asymptotic bias does not change with the number of iterations $K$ and the asymptotic variance is relatively negligible.

\begin{table}[H]
\begin{center}
\scalebox{0.85}
{
\begin{tabular}{cc|ccc|ccc|ccc}\hline
\hline
\multicolumn{1}{c}{\multirow{2}[4]{*}{$K$}} & \multicolumn{1}{c|}{\multirow{2}[4]{*}{Statistic}} & \multicolumn{3}{c|}{$K$-MD(${\bf I}_{|\tilde{A} \times X| \times |\tilde{A} \times X|}$)} & \multicolumn{3}{c|}{$K$-MD($W_{AV}^{*}$)} & \multicolumn{3}{c}{$K$-ML}\\
\multicolumn{1}{c}{} & \multicolumn{1}{c|}{} & 
\multicolumn{1}{c}{$n=200$} & \multicolumn{1}{c}{$n=500$} & \multicolumn{1}{c|}{$n=1,000$} & 
\multicolumn{1}{c}{$n=200$} & \multicolumn{1}{c}{$n=500$} & \multicolumn{1}{c|}{$n=1,000$} & 
\multicolumn{1}{c}{$n=200$} & \multicolumn{1}{c}{$n=500$} & \multicolumn{1}{c}{$n=1,000$} 
\\
\hline
 & $\sqrt{n}~$Bias & 0.07 & 0.02 & 0.01 & 0.06 & 0.02 & 0.01 & 0.05 & 0.02 & 0.01 \\
$1$ & $\sqrt{n}~$SD & 
0.24 & 0.25 & 0.24 & 0.23 & 0.23 & 0.22 & 0.22 & 0.22 & 0.22 \\
  & $n~$MSE &
0.06 & 0.06 & 0.06 & 0.06 & 0.05 & 0.05 & 0.05 & 0.05 & 0.05 \\
  \hline
   & $\sqrt{n}~$Bias & 
0.00 & 0.00 & 0.00 & 0.00 & 0.00 & 0.00 & 0.00 & 0.00 & 0.00 \\
$2$ & $\sqrt{n}~$SD &
0.24 & 0.24 & 0.24 & 0.23 & 0.23 & 0.22 & 0.22 & 0.22 & 0.22 \\
   & $n~$MSE & 
0.06 & 0.06 & 0.06 & 0.05 & 0.05 & 0.05 & 0.05 & 0.05 & 0.05 \\
\hline
 & $\sqrt{n}~$Bias & 
0.00 & 0.00 & 0.00 & 0.00 & 0.00 & 0.00 & 0.00 & 0.00 & 0.00 \\
$3$ & $\sqrt{n}~$SD & 
0.24 & 0.25 & 0.24 & 0.23 & 0.23 & 0.22 & 0.22 & 0.22 & 0.22 \\
   & $n~$MSE & 
0.06 & 0.06 & 0.06 & 0.05 & 0.05 & 0.05 & 0.05 & 0.05 & 0.05 \\
  \hline
   & $\sqrt{n}~$Bias &
0.00 & 0.00 & 0.00 & 0.00 & 0.00 & 0.00 & 0.00 & 0.00 & 0.00 \\
 $10$ & $\sqrt{n}~$SD & 
0.24 & 0.25 & 0.24 & 0.23 & 0.23 & 0.22 & 0.22 & 0.22 & 0.22 \\
   & $n~$MSE & 
0.06 & 0.06 & 0.06 & 0.05 & 0.05 & 0.05 & 0.05 & 0.05 & 0.05 \\
\hline
\hline
\end{tabular}
	}\end{center}
	\caption{Simulation results in the second misspecification design under local misspecification with $\tau_n \propto n^{-1}$.}
	\label{tab:miss1_2}
\end{table}

\begin{table}[H]
\begin{center}
\scalebox{0.85}
{
\begin{tabular}{cc|ccc|ccc|ccc}\hline
\hline
\multicolumn{1}{c}{\multirow{2}[4]{*}{$K$}} & \multicolumn{1}{c|}{\multirow{2}[4]{*}{Statistic}} & \multicolumn{3}{c|}{$K$-MD(${\bf I}_{|\tilde{A} \times X| \times |\tilde{A} \times X|}$)} & \multicolumn{3}{c|}{$K$-MD($W_{AV}^{*}$)} & \multicolumn{3}{c}{$K$-ML}\\
\multicolumn{1}{c}{} & \multicolumn{1}{c|}{} & 
\multicolumn{1}{c}{$n=200$} & \multicolumn{1}{c}{$n=500$} & \multicolumn{1}{c|}{$n=1,000$} & 
\multicolumn{1}{c}{$n=200$} & \multicolumn{1}{c}{$n=500$} & \multicolumn{1}{c|}{$n=1,000$} & 
\multicolumn{1}{c}{$n=200$} & \multicolumn{1}{c}{$n=500$} & \multicolumn{1}{c}{$n=1,000$} 
\\
\hline
  & $\sqrt{n}~$Bias & -0.04 & -0.08 & -0.09 & -0.05 & -0.08 & -0.09 & -0.05 & -0.08 & -0.09 \\
$1$ & $\sqrt{n}~$SD & 0.23 & 0.24 & 0.23 & 0.22 & 0.22 & 0.22 & 0.21 & 0.22 & 0.21 \\
  & $n~$MSE & 0.05 & 0.06 & 0.06 & 0.05 & 0.06 & 0.06 & 0.05 & 0.05 & 0.05 \\
  \hline
 & $\sqrt{n}~$Bias & 
-0.10 & -0.10 & -0.10 & -0.10 & -0.10 & -0.10 & -0.10 & -0.10 & -0.10 \\
$2$  & $\sqrt{n}~$SD & 
0.22 & 0.23 & 0.23 & 0.21 & 0.22 & 0.21 & 0.20 & 0.21 & 0.21 \\
   & $n~$MSE & 
0.06 & 0.06 & 0.06 & 0.05 & 0.06 & 0.06 & 0.05 & 0.05 & 0.06 \\
   \hline
 & $\sqrt{n}~$Bias & 
-0.10 & -0.10 & -0.10 & -0.10 & -0.10 & -0.10 & -0.10 & -0.10 & -0.10 \\
  $3$ & $\sqrt{n}~$SD & 
0.22 & 0.23 & 0.23 & 0.21 & 0.22 & 0.21 & 0.20 & 0.21 & 0.21 \\
   & $n~$MSE & 
0.06 & 0.06 & 0.06 & 0.05 & 0.06 & 0.06 & 0.05 & 0.05 & 0.06 \\
   \hline
   & $\sqrt{n}~$Bias & 
-0.10 & -0.10 & -0.10 & -0.10 & -0.10 & -0.10 & -0.10 & -0.10 & -0.10 \\
 $10$ & $\sqrt{n}~$SD & 
0.23 & 0.23 & 0.23 & 0.21 & 0.22 & 0.21 & 0.20 & 0.21 & 0.21 \\
   & $n~$MSE &
0.06 & 0.06 & 0.06 & 0.05 & 0.06 & 0.06 & 0.05 & 0.05 & 0.06 \\
\hline
\hline
\end{tabular}
	}\end{center}
	\caption{Simulation results in the second misspecification design under local misspecification with $\tau_n \propto  n^{-1/2}$.}
	\label{tab:miss2_2}
\end{table}

\begin{table}[H]
\begin{center}
\scalebox{0.85}
{
\begin{tabular}{cc|ccc|ccc|ccc}\hline
\hline
\multicolumn{1}{c}{\multirow{2}[4]{*}{$K$}} & \multicolumn{1}{c|}{\multirow{2}[4]{*}{Statistic}} & \multicolumn{3}{c|}{$K$-MD(${\bf I}_{|\tilde{A} \times X| \times |\tilde{A} \times X|}$)} & \multicolumn{3}{c|}{$K$-MD($W_{AV}^{*}$)} & \multicolumn{3}{c}{$K$-ML}\\
\multicolumn{1}{c}{} & \multicolumn{1}{c|}{} & 
\multicolumn{1}{c}{$n=200$} & \multicolumn{1}{c}{$n=500$} & \multicolumn{1}{c|}{$n=1,000$} & 
\multicolumn{1}{c}{$n=200$} & \multicolumn{1}{c}{$n=500$} & \multicolumn{1}{c|}{$n=1,000$} & 
\multicolumn{1}{c}{$n=200$} & \multicolumn{1}{c}{$n=500$} & \multicolumn{1}{c}{$n=1,000$} 
\\
\hline
  & $n^{1/3}~$Bias &-0.07 & -0.09 & -0.09 & -0.07 & -0.09 & -0.09 & -0.07 & -0.09 & -0.09 \\
$1$ & $n^{1/3}~$SD & 0.08 & 0.08 & 0.07 & 0.08 & 0.07 & 0.06 & 0.08 & 0.07 & 0.06 \\
   & $n^{2/3}~$MSE &
0.01 & 0.01 & 0.01 & 0.01 & 0.01 & 0.01 & 0.01 & 0.01 & 0.01 \\
   \hline
  & $n^{1/3}~$Bias & 
-0.09 & -0.09 & -0.10 & -0.09 & -0.09 & -0.10 & -0.09 & -0.09 & -0.10 \\
$2$ & $n^{1/3}~$SD & 
0.08 & 0.08 & 0.07 & 0.08 & 0.07 & 0.06 & 0.07 & 0.07 & 0.06 \\
   & $n^{2/3}~$MSE & 
0.02 & 0.01 & 0.01 & 0.01 & 0.01 & 0.01 & 0.01 & 0.01 & 0.01 \\
   \hline
  & $n^{1/3}~$Bias & 
-0.09 & -0.09 & -0.10 & -0.09 & -0.09 & -0.10 & -0.09 & -0.09 & -0.10 \\
$3$ & $n^{1/3}~$SD & 
0.08 & 0.08 & 0.07 & 0.08 & 0.07 & 0.06 & 0.07 & 0.07 & 0.06 \\
   & $n^{2/3}~$MSE & 
0.02 & 0.01 & 0.01 & 0.01 & 0.01 & 0.01 & 0.01 & 0.01 & 0.01 \\
   \hline
 & $n^{1/3}~$Bias & 
-0.09 & -0.09 & -0.10 & -0.09 & -0.09 & -0.10 & -0.09 & -0.09 & -0.10 \\
$10$ & $n^{1/3}~$SD & 
0.08 & 0.08 & 0.07 & 0.08 & 0.07 & 0.06 & 0.07 & 0.07 & 0.06 \\
   & $n^{2/3}~$MSE & 
0.02 & 0.01 & 0.01 & 0.01 & 0.01 & 0.01 & 0.01 & 0.01 & 0.01 \\
\hline
\hline
\end{tabular}
	}\end{center}
	\caption{Simulation results in the second misspecification design under local misspecification with $\tau_n \propto  n^{-1/3}$ and using the correct scaling.}
	\label{tab:miss3b_2}
\end{table}

\begin{table}[H]
\begin{center}
\scalebox{0.85}
{
\begin{tabular}{cc|ccc|ccc|ccc}\hline
\hline
\multicolumn{1}{c}{\multirow{2}[4]{*}{$K$}} & \multicolumn{1}{c|}{\multirow{2}[4]{*}{Statistic}} & \multicolumn{3}{c|}{$K$-MD(${\bf I}_{|\tilde{A} \times X| \times |\tilde{A} \times X|}$)} & \multicolumn{3}{c|}{$K$-MD($W_{AV}^{*}$)} & \multicolumn{3}{c}{$K$-ML}\\
\multicolumn{1}{c}{} & \multicolumn{1}{c|}{} & 
\multicolumn{1}{c}{$n=200$} & \multicolumn{1}{c}{$n=500$} & \multicolumn{1}{c|}{$n=1,000$} & 
\multicolumn{1}{c}{$n=200$} & \multicolumn{1}{c}{$n=500$} & \multicolumn{1}{c|}{$n=1,000$} & 
\multicolumn{1}{c}{$n=200$} & \multicolumn{1}{c}{$n=500$} & \multicolumn{1}{c}{$n=1,000$} 
\\
\hline
  & $\sqrt{n}~$Bias &-0.17 & -0.25 & -0.30 & -0.17 & -0.25 & -0.30 & -0.17 & -0.25 & -0.30 \\
$1$ & $\sqrt{n}~$SD & 
0.20 & 0.22 & 0.22 & 0.20 & 0.20 & 0.20 & 0.19 & 0.20 & 0.20 \\
   & $n~$MSE & 
0.07 & 0.11 & 0.13 & 0.07 & 0.10 & 0.13 & 0.07 & 0.10 & 0.13 \\
   \hline
    & $\sqrt{n}~$Bias &
-0.23 & -0.26 & -0.30 & -0.22 & -0.26 & -0.30 & -0.22 & -0.26 & -0.30 \\
 $2$ & $\sqrt{n}~$SD & 
0.20 & 0.21 & 0.21 & 0.19 & 0.20 & 0.20 & 0.18 & 0.19 & 0.20 \\
   & $n~$MSE & 
0.09 & 0.12 & 0.14 & 0.09 & 0.11 & 0.13 & 0.08 & 0.11 & 0.13 \\
   \hline
  & $\sqrt{n}~$Bias & 
-0.23 & -0.27 & -0.30 & -0.22 & -0.26 & -0.30 & -0.22 & -0.26 & -0.30 \\
$3$ & $\sqrt{n}~$SD & 
0.20 & 0.21 & 0.21 & 0.19 & 0.20 & 0.20 & 0.18 & 0.19 & 0.20 \\
   & $n~$MSE & 
0.09 & 0.12 & 0.14 & 0.09 & 0.11 & 0.13 & 0.08 & 0.11 & 0.13 \\
   \hline
 & $\sqrt{n}~$Bias & 
-0.23 & -0.27 & -0.30 & -0.22 & -0.26 & -0.30 & -0.22 & -0.26 & -0.30 \\
$10$ & $\sqrt{n}~$SD & 
0.20 & 0.21 & 0.21 & 0.19 & 0.20 & 0.20 & 0.18 & 0.19 & 0.20 \\
   & $n~$MSE & 
0.09 & 0.12 & 0.14 & 0.09 & 0.11 & 0.13 & 0.08 & 0.11 & 0.13 \\
\hline
\hline
\end{tabular}
	}\end{center}
	\caption{Simulation results in the second misspecification design under local misspecification with $\tau_n \propto  n^{-1/3}$ and using the regular scaling (i.e.\ $\sqrt{n}$).}
	\label{tab:miss3a_2}
\end{table}

\subsection{Simulation results for the third misspecification design}

This section describes Monte Carlo simulation results for a third misspecification design. Once again, the econometric model is exactly as the one described in Section \ref{sec:MonteCarlos}. The true DGP is as in the third illustration in Section \ref{sec:Misspecification}, and considers agents that depart from rational behavior.

Given state variables $(x,\epsilon)$, our model predicts that agents choose the action that maximizes the expected discounted utility. Instead, we simulate agents who make choices according to a multinomial distribution with choice probabilities that are increasing in the action-specific expected discounted utility. Specifically, we use the choice probabilities in Eq.\ \eqref{eq:Irrational} for some $\tau_n >0$.

The econometric model model is then misspecified in the sense that it presumes rationality, i.e., $\tau_n \to 0$. We impose local misspecification by setting $\tau_n \equiv 10 n^{-\delta}$ with $\delta \in \{1/3,1/2,1\}$. The rest of the parameters used to implement the Monte Carlo simulation are exactly as in Section \ref{sec:MonteCarlos}.

For the sake of comparison with previous simulation designs, we present results for the estimator of $\theta_{u,2}$. Once again, the simulation results are qualitatively similar to the ones obtained in the previous section and support all of our theoretical conclusions. The results in Tables \ref{tab:miss1_3},  \ref{tab:miss2_3}, \ref{tab:miss3b_3}, and \ref{tab:miss3a_3} are analogous to Tables \ref{tab:miss1},  \ref{tab:miss2}, \ref{tab:miss3b}, and \ref{tab:miss3a}, respectively. We refer to Section \ref{sec:MonteCarlos} for a description of these results.

\begin{table}[H]
\begin{center}
\scalebox{0.85}
{
\begin{tabular}{cc|ccc|ccc|ccc}\hline
\hline
\multicolumn{1}{c}{\multirow{2}[4]{*}{$K$}} & \multicolumn{1}{c|}{\multirow{2}[4]{*}{Statistic}} & \multicolumn{3}{c|}{$K$-MD(${\bf I}_{|\tilde{A} \times X| \times |\tilde{A} \times X|}$)} & \multicolumn{3}{c|}{$K$-MD($W_{AV}^{*}$)} & \multicolumn{3}{c}{$K$-ML}\\
\multicolumn{1}{c}{} & \multicolumn{1}{c|}{} & 
\multicolumn{1}{c}{$n=200$} & \multicolumn{1}{c}{$n=500$} & \multicolumn{1}{c|}{$n=1,000$} & 
\multicolumn{1}{c}{$n=200$} & \multicolumn{1}{c}{$n=500$} & \multicolumn{1}{c|}{$n=1,000$} & 
\multicolumn{1}{c}{$n=200$} & \multicolumn{1}{c}{$n=500$} & \multicolumn{1}{c}{$n=1,000$} 
\\
\hline
& $\sqrt{n}~$Bias & 0.07 & 0.03 & 0.01 & 0.06 & 0.02 & 0.01 & 0.06 & 0.02 & 0.01 \\
 $1$ & $\sqrt{n}~$SD & 
0.25 & 0.25 & 0.24 & 0.24 & 0.23 & 0.22 & 0.22 & 0.23 & 0.22 \\
   & $n~$MSE &
0.07 & 0.06 & 0.06 & 0.06 & 0.05 & 0.05 & 0.05 & 0.05 & 0.05 \\
   \hline
    & $\sqrt{n}~$Bias & 
0.00 & 0.01 & 0.00 & 0.00 & 0.01 & 0.00 & 0.01 & 0.01 & 0.00 \\
$2$ & $\sqrt{n}~$SD & 
0.24 & 0.25 & 0.24 & 0.23 & 0.23 & 0.22 & 0.22 & 0.22 & 0.22 \\
   & $n~$MSE & 
0.06 & 0.06 & 0.06 & 0.05 & 0.05 & 0.05 & 0.05 & 0.05 & 0.05 \\
 \hline & $\sqrt{n}~$Bias & 
0.00 & 0.01 & 0.00 & 0.00 & 0.01 & 0.00 & 0.00 & 0.01 & 0.00 \\
$3$ & $\sqrt{n}~$SD & 
0.25 & 0.25 & 0.24 & 0.23 & 0.23 & 0.22 & 0.22 & 0.22 & 0.22 \\
   & $n~$MSE & 
0.06 & 0.06 & 0.06 & 0.05 & 0.05 & 0.05 & 0.05 & 0.05 & 0.05 \\
   \hline
   & $\sqrt{n}~$Bias &
0.00 & 0.01 & 0.00 & 0.00 & 0.01 & 0.00 & 0.00 & 0.01 & 0.00 \\
 $10$ & $\sqrt{n}~$SD & 
0.25 & 0.25 & 0.24 & 0.23 & 0.23 & 0.22 & 0.22 & 0.22 & 0.22 \\
   & $n~$MSE & 
0.06 & 0.06 & 0.06 & 0.05 & 0.05 & 0.05 & 0.05 & 0.05 & 0.05 \\
\hline
\hline
\end{tabular}
	}\end{center}
	\caption{Simulation results in the third misspecification design under local misspecification with $\tau_n \propto n^{-1}$.}
	\label{tab:miss1_3}
\end{table}

\begin{table}[H]
\begin{center}
\scalebox{0.85}
{
\begin{tabular}{cc|ccc|ccc|ccc}\hline
\hline
\multicolumn{1}{c}{\multirow{2}[4]{*}{$K$}} & \multicolumn{1}{c|}{\multirow{2}[4]{*}{Statistic}} & \multicolumn{3}{c|}{$K$-MD(${\bf I}_{|\tilde{A} \times X| \times |\tilde{A} \times X|}$)} & \multicolumn{3}{c|}{$K$-MD($W_{AV}^{*}$)} & \multicolumn{3}{c}{$K$-ML}\\
\multicolumn{1}{c}{} & \multicolumn{1}{c|}{} & 
\multicolumn{1}{c}{$n=200$} & \multicolumn{1}{c}{$n=500$} & \multicolumn{1}{c|}{$n=1,000$} & 
\multicolumn{1}{c}{$n=200$} & \multicolumn{1}{c}{$n=500$} & \multicolumn{1}{c|}{$n=1,000$} & 
\multicolumn{1}{c}{$n=200$} & \multicolumn{1}{c}{$n=500$} & \multicolumn{1}{c}{$n=1,000$} 
\\
\hline
  & $\sqrt{n}~$Bias & -0.07 & -0.12 & -0.10 & -0.09 & -0.12 & -0.11 & -0.09 & -0.12 & -0.11 \\
$1$ & $\sqrt{n}~$SD & 
0.25 & 0.24 & 0.24 & 0.23 & 0.22 & 0.22 & 0.22 & 0.22 & 0.22 \\
  & $n~$MSE & 
0.07 & 0.07 & 0.07 & 0.06 & 0.07 & 0.06 & 0.06 & 0.06 & 0.06 \\
  \hline
 & $\sqrt{n}~$Bias & 
-0.15 & -0.13 & -0.11 & -0.16 & -0.14 & -0.11 & -0.16 & -0.14 & -0.11 \\
$2$ & $\sqrt{n}~$SD & 
0.23 & 0.24 & 0.24 & 0.22 & 0.22 & 0.22 & 0.21 & 0.21 & 0.22 \\
   & $n~$MSE & 
0.08 & 0.07 & 0.07 & 0.07 & 0.07 & 0.06 & 0.07 & 0.07 & 0.06 \\
   \hline
 & $\sqrt{n}~$Bias & 
-0.16 & -0.13 & -0.11 & -0.16 & -0.14 & -0.11 & -0.16 & -0.14 & -0.11 \\
  $3$ & $\sqrt{n}~$SD & 
0.24 & 0.24 & 0.24 & 0.22 & 0.22 & 0.22 & 0.21 & 0.22 & 0.22 \\
   & $n~$MSE & 
0.08 & 0.07 & 0.07 & 0.07 & 0.07 & 0.06 & 0.07 & 0.07 & 0.06 \\
   \hline
   & $\sqrt{n}~$Bias &
-0.16 & -0.13 & -0.11 & -0.16 & -0.14 & -0.11 & -0.16 & -0.14 & -0.11 \\
 $10$ & $\sqrt{n}~$SD & 
0.24 & 0.24 & 0.24 & 0.22 & 0.22 & 0.22 & 0.21 & 0.22 & 0.22 \\
   & $n~$MSE & 
0.08 & 0.07 & 0.07 & 0.08 & 0.07 & 0.06 & 0.07 & 0.07 & 0.06 \\
\hline
\hline
\end{tabular}
	}\end{center}
	\caption{Simulation results in the third misspecification design under local misspecification with $\tau_n \propto  n^{-1/2}$.}
	\label{tab:miss2_3}
\end{table}

\begin{table}[H]
\begin{center}
\scalebox{0.85}
{
\begin{tabular}{cc|ccc|ccc|ccc}\hline
\hline
\multicolumn{1}{c}{\multirow{2}[4]{*}{$K$}} & \multicolumn{1}{c|}{\multirow{2}[4]{*}{Statistic}} & \multicolumn{3}{c|}{$K$-MD(${\bf I}_{|\tilde{A} \times X| \times |\tilde{A} \times X|}$)} & \multicolumn{3}{c|}{$K$-MD($W_{AV}^{*}$)} & \multicolumn{3}{c}{$K$-ML}\\
\multicolumn{1}{c}{} & \multicolumn{1}{c|}{} & 
\multicolumn{1}{c}{$n=200$} & \multicolumn{1}{c}{$n=500$} & \multicolumn{1}{c|}{$n=1,000$} & 
\multicolumn{1}{c}{$n=200$} & \multicolumn{1}{c}{$n=500$} & \multicolumn{1}{c|}{$n=1,000$} & 
\multicolumn{1}{c}{$n=200$} & \multicolumn{1}{c}{$n=500$} & \multicolumn{1}{c}{$n=1,000$} 
\\
\hline
  & $n^{1/3}~$Bias & -0.10 & -0.16 & -0.17 & -0.11 & -0.16 & -0.17 & -0.11 & -0.16 & -0.17 \\
$1$ & $n^{1/3}~$SD &
0.10 & 0.08 & 0.07 & 0.10 & 0.08 & 0.07 & 0.09 & 0.08 & 0.07 \\
   & $n^{2/3}~$MSE & 
0.02 & 0.03 & 0.03 & 0.02 & 0.03 & 0.03 & 0.02 & 0.03 & 0.03 \\
\hline
  & $n^{1/3}~$Bias & 
-0.14 & -0.17 & -0.17 & -0.15 & -0.17 & -0.17 & -0.15 & -0.17 & -0.17 \\
  $2$ & $n^{1/3}~$SD & 
0.09 & 0.08 & 0.07 & 0.08 & 0.08 & 0.07 & 0.08 & 0.07 & 0.07 \\
   & $n^{2/3}~$MSE & 
0.03 & 0.03 & 0.03 & 0.03 & 0.03 & 0.04 & 0.03 & 0.03 & 0.03 \\
   \hline
  & $n^{1/3}~$Bias & 
-0.14 & -0.17 & -0.17 & -0.15 & -0.17 & -0.17 & -0.15 & -0.17 & -0.17 \\
$3$ & $n^{1/3}~$SD & 
0.09 & 0.08 & 0.07 & 0.08 & 0.08 & 0.07 & 0.08 & 0.07 & 0.07 \\
   & $n^{2/3}~$MSE & 
0.03 & 0.03 & 0.03 & 0.03 & 0.03 & 0.04 & 0.03 & 0.03 & 0.03 \\
   \hline
 & $n^{1/3}~$Bias & 
-0.14 & -0.17 & -0.17 & -0.15 & -0.17 & -0.17 & -0.15 & -0.17 & -0.17 \\
$10$ & $n^{1/3}~$SD & 
0.09 & 0.08 & 0.07 & 0.08 & 0.08 & 0.07 & 0.08 & 0.07 & 0.07 \\
   & $n^{2/3}~$MSE & 
0.03 & 0.03 & 0.03 & 0.03 & 0.03 & 0.04 & 0.03 & 0.03 & 0.03 \\
\hline
\hline
\end{tabular}
	}\end{center}
	\caption{Simulation results in the third misspecification design under local misspecification with $\tau_n \propto  n^{-1/3}$ and using the correct scaling.}
	\label{tab:miss3b_3}
\end{table}

\begin{table}[H]
\begin{center}
\scalebox{0.85}
{
\begin{tabular}{cc|ccc|ccc|ccc}\hline
\hline
\multicolumn{1}{c}{\multirow{2}[4]{*}{$K$}} & \multicolumn{1}{c|}{\multirow{2}[4]{*}{Statistic}} & \multicolumn{3}{c|}{$K$-MD(${\bf I}_{|\tilde{A} \times X| \times |\tilde{A} \times X|}$)} & \multicolumn{3}{c|}{$K$-MD($W_{AV}^{*}$)} & \multicolumn{3}{c}{$K$-ML}\\
\multicolumn{1}{c}{} & \multicolumn{1}{c|}{} & 
\multicolumn{1}{c}{$n=200$} & \multicolumn{1}{c}{$n=500$} & \multicolumn{1}{c|}{$n=1,000$} & 
\multicolumn{1}{c}{$n=200$} & \multicolumn{1}{c}{$n=500$} & \multicolumn{1}{c|}{$n=1,000$} & 
\multicolumn{1}{c}{$n=200$} & \multicolumn{1}{c}{$n=500$} & \multicolumn{1}{c}{$n=1,000$} 
\\
\hline
  & $\sqrt{n}~$Bias &-0.25 & -0.45 & -0.53 & -0.26 & -0.46 & -0.54 & -0.27 & -0.46 & -0.54 \\
$1$ & $\sqrt{n}~$SD & 
0.25 & 0.24 & 0.23 & 0.23 & 0.22 & 0.22 & 0.22 & 0.21 & 0.21 \\
   & $n~$MSE & 
0.12 & 0.25 & 0.34 & 0.12 & 0.26 & 0.34 & 0.12 & 0.26 & 0.34 \\
   \hline
    & $\sqrt{n}~$Bias & 
-0.34 & -0.47 & -0.54 & -0.35 & -0.48 & -0.55 & -0.36 & -0.48 & -0.55 \\
 $2$ & $\sqrt{n}~$SD & 
0.21 & 0.23 & 0.23 & 0.19 & 0.21 & 0.21 & 0.19 & 0.21 & 0.21 \\
   & $n~$MSE & 
0.16 & 0.27 & 0.35 & 0.16 & 0.27 & 0.35 & 0.16 & 0.27 & 0.35 \\
   \hline
  & $\sqrt{n}~$Bias & 
-0.35 & -0.47 & -0.54 & -0.36 & -0.48 & -0.55 & -0.36 & -0.48 & -0.55 \\
$3$ & $\sqrt{n}~$SD &
0.21 & 0.23 & 0.23 & 0.20 & 0.21 & 0.21 & 0.19 & 0.21 & 0.21 \\
   & $n~$MSE &
0.17 & 0.27 & 0.35 & 0.17 & 0.27 & 0.35 & 0.17 & 0.27 & 0.35 \\
   \hline
 & $\sqrt{n}~$Bias & 
-0.35 & -0.47 & -0.54 & -0.36 & -0.48 & -0.55 & -0.37 & -0.48 & -0.55 \\
$10$ & $\sqrt{n}~$SD & 
0.21 & 0.23 & 0.23 & 0.20 & 0.21 & 0.21 & 0.19 & 0.21 & 0.21 \\
   & $n~$MSE & 
0.17 & 0.27 & 0.35 & 0.17 & 0.27 & 0.35 & 0.17 & 0.27 & 0.35 \\
\hline
\hline
\end{tabular}
	}\end{center}
	\caption{Simulation results in the third misspecification design under local misspecification with $\tau_n \propto  n^{-1/3}$ and using the regular scaling (i.e.\ $\sqrt{n}$).}
	\label{tab:miss3a_3}
\end{table}

\end{document}